\RequirePackage{fix-cm}

\documentclass[smallextended]{svjour3} 

\smartqed  

\AtBeginDocument{%
  \paperwidth=\dimexpr
    1in + \oddsidemargin
    + \textwidth
    + 1in + \oddsidemargin
  \relax
  \paperheight=\dimexpr
    1in + \topmargin
    + \headheight + \headsep
    + \textheight
    + 1in + \topmargin
  \relax
  \usepackage[pass]{geometry}\relax
}

\usepackage[T1]{fontenc}
\usepackage[utf8]{inputenc}
\usepackage{lmodern}

\usepackage{color}
\usepackage{longtable}

\usepackage[cmex10]{amsmath}
\usepackage{amssymb,amsfonts}
\usepackage{wasysym}

\usepackage{graphicx}
\usepackage{algorithmic}
\usepackage{array}
\usepackage[caption=false,font=footnotesize]{subfig}
\usepackage{url}
\usepackage{enumitem}
\usepackage{multicol}
\usepackage{caption}
\usepackage[bb=boondox]{mathalfa}

\newcommand{\node}[1]{\mathrm{#1}}
\newcommand{\dist}{\mathrm{d}}

\newcommand{\ds}{\displaystyle}
\newcommand{\ts}{\textstyle}

\newcommand{\Rr}{{\mathbb R}}

\newcommand{\Ll}{\mathbf{L}}
\newcommand{\Aa}{\mathbf{A}}
\newcommand{\Dd}{\mathbf{D}}
\newcommand{\Kk}{\mathbf{K}}
\newcommand{\Mm}{\mathbf{M}}
\newcommand{\Uu}{\mathbf{U}}
\newcommand{\Ccc}{\mathbf{C}}

\newcommand{\itpx}{x_{\circ}}

\usepackage{pgfplots}

\usepackage{tikz}
\usetikzlibrary{calc,fadings,decorations.pathreplacing}
\usetikzlibrary{matrix,arrows}
\usetikzlibrary{positioning,fit,shapes,external,3d,patterns,spy}
\usetikzlibrary[shapes.arrows]
\usetikzlibrary{pgfplots.colormaps} 

\usepackage[ruled,vlined]{algorithm2e}

\newcommand{\1}{\mathrm{1\hspace{-1mm}l}}

\newcommand{\minus}{%
  \setbox0=\hbox{-}%
  \vcenter{%
    \hrule width\wd0 height \the\fontdimen8\textfont3%
  }%
}

\newcommand{\normp}[1]{\raisebox{-2pt}[0mm][0mm]{\scalebox{.55}{$\mathcal{L}^{p}(#1)$}}}
\newcommand{\norminf}[1]{\raisebox{-2pt}[0mm][0mm]{\scalebox{.55}{$\mathcal{L}^{\infty}(#1)$}}}
\newcommand{\norminfq}[1]{\raisebox{-2pt}[0mm][0mm]{\scalebox{.55}{$\mathcal{L}^{\infty,q}(#1)$}}}
\newcommand{\norminfp}[1]{\raisebox{-2pt}[0mm][0mm]{\scalebox{.55}{$\mathcal{L}^{\infty,p}(#1)$}}}

\newcommand{\subind}[2]{\raisebox{-0.5pt}[0mm][0mm]{\scalebox{.55}{$#1,\hspace{-1pt} #2$}}}

\newcommand{\error}{\mathcal{E}}

\spnewtheorem{thm}{Theorem}{\bfseries}{\itshape}
\spnewtheorem{pro}[thm]{Proposition}{\bfseries}{\itshape}
\spnewtheorem{cor}{Corollary}{\bfseries}{\itshape}
\spnewtheorem{lem}{Lemma}{\bfseries}{\itshape}

\spnewtheorem{dfn}{Definition}{\bfseries}{\rmfamily}
\spnewtheorem{rem}{Remark}{\bfseries}{\rmfamily}
\spnewtheorem{exa}{Example}{\bfseries}{\rmfamily}

\begin{document}

\title{Partition of Unity Methods for Signal Processing on Graphs}
\subtitle{}

\author{Roberto Cavoretto \and
        Alessandra \mbox{De Rossi}  \and
        Wolfgang Erb       
}

\institute{R. Cavoretto, A. De Rossi \at
              University of Torino \\
              \email{roberto.cavoretto@unito.it} \\         
              \email{alessandra.derossi@unito.it} \\
	       W. Erb \at University of Padova \\
           \email{erb@math.unipd.it} 
}

\date{\today}

\titlerunning{Partition of Unity Methods for Signal Processing on Graphs}
\authorrunning{R. Cavoretto, A. De Rossi, W. Erb }

\maketitle

\begin{abstract} Partition of unity methods (PUMs) on graphs are simple and highly adaptive auxiliary tools for graph signal processing. Based on a greedy-type metric clustering and augmentation scheme, we show how a partition of unity can be generated in an efficient way on graphs. We investigate how PUMs can be combined with a local graph basis function (GBF) approximation method in order to obtain low-cost global interpolation or classification schemes. From a theoretical point of view, we study necessary prerequisites for the partition of unity such that global error estimates of the PUM follow from corresponding local ones. Finally, properties of the PUM as cost-efficiency and approximation accuracy are investigated numerically. 

\end{abstract}

\section{Introduction}

\begin{figure}[htbp]
	\centering
 	\includegraphics[width= 1\textwidth]{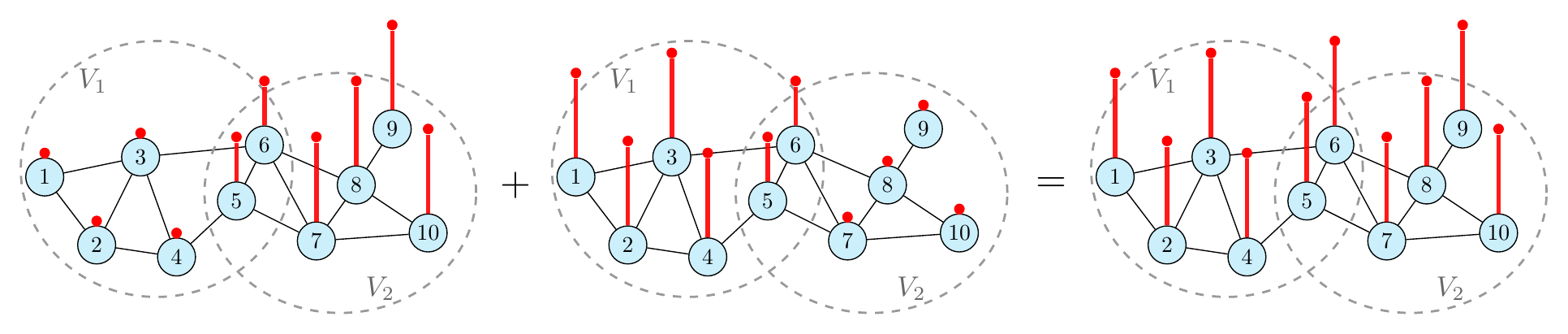} 
	\caption{Schematic sketch of a partition of unity on a graph.}
	\label{fig:graphPUM}
\end{figure}

Graph signal processing is a cutting-edge research field for the study of graph signals in which mathematical processing tools as filtering, compression, noise removal, sampling, or decomposition methods are investigated \cite{Ortega2018,Pesenson2008,shuman2016,StankovicDakovicSejdic2019}. Graph structures appear naturally in a multitude of modern applications, as in social networks, traffic maps or biological networks. In general, these networks exhibit a large number of vertices and a highly irregular edge structure. In order to be able to deal with signals on such irregular graphs, efficient and fast processing tools are necessary.

In this work, we introduce and analyze partition of unity methods (PUMs) on graphs as flexible and efficient auxiliary tools to individualize and accelerate signal processing steps. Many algorithms in graph signal processing as, for instance, the calculation of the graph Fourier transform get computationally infeasable if the size of the graph is too large or the topological structure of the network is not sparse. PUMs allow, in an efficient way, to perform operations as signal reconstruction from samples, classification of nodes, or signal filtering locally on smaller portions of the graph, and, then, to rebuild the global signal from the local ones (see Fig. \ref{fig:graphPUM} for a simple sketch). This makes a PUM to an ideal auxiliary tool also if more adaptivity is required and processing steps have to be individualized to local spatial prerequisites on the graph.

In recent years, PUMs have been successfully combined with a multitude of different computational methods. In \cite{GriebelSchweitzer2000,Melenk1996}, for instance, PUMs have been considered in the context of meshfree Galerkin methods to obtain more adaptive and robust solvers for differential equations. In the approximation with radial basis functions (RBFs), the combination of RBFs with PUMs yields significantly sparser system matrices in collocation or interpolation problems, and, therefore, a considerable speed-up of calculations \cite{cavoretto2015,cavoretto2018,cavoretto2019,cavoretto2020}, \cite[Chapter 29]{Fasshauer2007}, \cite{Larsson2017,Wendland2002}, \cite[Section 15.4]{We05}.

On graphs, PUMs combined with kernel-based interpolation or approximation techniques have similar benefits: the computations can be reduced, in a highly-flexible way, to local calculations on smaller portions of the graph. In this way, expensive global costs can be avoided and parameters of the approximating scheme can be adjusted to local prerequisites. In this work, we will focus on approximation methods based on generalized translates of a graph basis function (GBF) (see \cite{erb2019b,erb2020} for an overview). This kernel-based method is an analog of RBF approximation in the euclidean space and contains diffusion kernels \cite{KondorLafferty2002} and variational splines \cite{Pesenson2009,Ward2018interpolating} on graphs as special cases.

For the construction of partition of unities, the discrete graph structure allows to take advantage of a multitude of clustering algorithms to generate, in an automatic way, a cover of the vertex domain by basic subdomains. The spectral structure induced by the graph Laplacian offers one possibility for such a clustering \cite{vonLuxburg2007}. In this article, we pursue a second possibility that uses a discrete metric distance on the graph to determine $J$ center nodes for the clusters \cite{Gonzalez1985,hochbaumshmoys1985}. This metric $J$-center clustering has the advantage that, with a slight modification of the algorithm, the selection of the centers can be restricted to nodes that are relevant for further computations. Once the clusters are determined, an augmentation procedure can be applied to the clusters in order to obtain an overlapping cover of subdomains. As soon as this cover is determined, the partition of unity itself can be obtained by the use of Shepard weight functions \cite{shepard1968}. 

As for PUMs in euclidean settings, overlapping subdomains are generally desirable also on graph domains. This allows to incorporate the global geometric information of the graph into the calculations. Further, overlapping domains allow to construct partition of unities with coherent boundary conditions. In Theorem \ref{thm:globalerrorestimate}, the main theoretical finding of this article, we will show that global approximants based on the PUM inherit smoothness properties of the local approximants on the subdomains once the partition of unity satisfies suitable boundary conditions. \\

\noindent \textbf{Outline and main contributions of this article} 
\begin{enumerate}
\item In Section \ref{sec:preliminaries}, we give a brief overview on spectral graph theory, the graph Fourier transform and approximation with GBFs.
\item An efficient construction of a partition of unity on metric graphs is provided in Section \ref{sec:PU}. In order to obtain a basic splitting of a graph in $J$ parts we will use an adaption of metric $J$-center clustering. This adaption will guarantee that every cluster contains at least one node with sampling information. In a second step, we augment the $J$ clusters with neighboring nodes to generate a cover of $J$ overlapping subdomains. Subordinate to this cover, we show how to generate a partition of unity on the graph.
\item In Algorithm \ref{algorithm1} (given in Section \ref{sec:four}), we combine the partition of unity constructed in Section \ref{sec:PU} with the GBF approximation scheme outlined in Section \ref{sec:GBF}. In this way, we are able to merge local GBF approximations on subgraphs to a global GBF-PUM approximation scheme on the graph.
\item In Section \ref{sec:theory}, we prove that, under particular assumptions on the partition of unity, smoothness of the global PUM approximants is inherited from local smoothness. The corresponding main result is stated as a global error estimate for PUM approximation in Theorem \ref{thm:globalerrorestimate}. 
\item In the final Section \ref{sec:numerics}, we conclude this article with some numerical experiments on the cost-efficiency and the accuracy of the PUM. 
\end{enumerate}

\section{Spectral graph theory and approximation with GBFs} \label{sec:preliminaries}

\subsection{Preliminaries on graph theory} \label{sec:spectralgraphtheory}

In this article, we will consider simple and connected graphs equipped with a graph Laplacian (providing a harmonic structure on the graph), and a metric distance. More precisely, we define a graph $G$ as a quadruplet $G=(V,E,\mathbf{L},\mathrm{d})$ with the following components:
\begin{itemize}
\item A vertex set $V=\{\node{v}_1, \ldots, \node{v}_{n}\}$ consisting of $n$ graph nodes. 
\item A set $E \subseteq V \times V$ containing all (directed) vertices $e_{\subind{i}{i'}} = (\node{v}_i, \node{v}_{i'})$, $i \neq i'$, of the graph. We assume that $e_{\subind{i}{i'}} \in E \; \Leftrightarrow \; e_{\subind{i'}{i}} \in E$, i.e., that the graph is undirected. 
\item A general symmetric graph Laplacian $\Ll \in \Rr^{n \times n}$ of the form (as, for instance, described in \cite[Section 13.9]{GodsilRoyle2001})
\begin{equation} \label{eq:generalizedLaplacian}
\ds {\begin{array}{ll}\; \Ll_{\subind{i}{i'}}<0& \text{if $i \neq i'$ and $\node{v}_{i}, \node{v}_{i'}$ are connected, i.e., $e_{\subind{i}{i'}} \in E$}, \\ \; \Ll_{\subind{i}{i'}}=0 & \text{if $i\neq i' $ and $\node{v}_{i}, \node{v}_{i'} $ are not connected}, \\ \; \Ll_{\subind{i}{i}} \in \Rr & \text{for $i \in \{1, \ldots, n\}$}.\end{array}}
\end{equation}

\item A symmetric distance metric $\mathrm{d}$ on the vertex set $V$, satisfying a triangle inequality. A typical example of such a metric is given by the graph geodesic, i.e., the length of the shortest path connecting two graph nodes. We assume that $G$ is a connected graph and, thus, that the distance $\mathrm{d}$ between two arbitrary nodes is finite. 
\end{itemize}

The negative non-diagonal elements $\Ll_{\subind{i}{i'}}$, $i \neq i'$, of the Laplacian $\Ll$ contain the connection weights of the edges $e_{\subind{i}{i'}} \in E$. These are usually collected as positive weights in the symmetric adjacency matrix $\Aa$:
\begin{equation*}
    \mathbf{A}_{\subind{i}{i'}} := 
  \begin{cases}
    - \Ll_{\subind{i}{i'}}, & \text{if $e_{\subind{i}{i'}} \in E$}, \\
    0, & \text{otherwise}.
  \end{cases}
\end{equation*}
In this article, we set no restrictions on the diagonal entries of the graph Laplacian $\Ll$. There are however some well-known examples in which the diagonal entries of the graph Laplacian are fixed. One example is the negative adjacency matrix $\Ll_A = - \Aa$ in which all diagonal entries of $\Ll_A$ vanish. Another more prominent example is the standard graph Laplacian 
\begin{equation} \label{eq:standardLaplacian} \Ll_S = \mathbf{D} - \Aa\end{equation}
in which the diagonal is determined by the degree matrix $\Dd$ given by 
\begin{equation*}
    \mathbf{D}_{\subind{i}{i'}} := 
  \begin{cases}
    \sum_{j=1}^n \mathbf{A}_{\subind{i}{j}}, & \text{if } i=i', \\
    0, & \text{otherwise}.
  \end{cases}
\end{equation*}
It is well-known that the standard Laplacian $\Ll_S$ in \eqref{eq:standardLaplacian} is a positive semi-definite matrix and the multiplicity of the eigenvalue $\lambda_1 = 0$ corresponds to the number of connected components of the graph $G$. A further important example is the normalized graph Laplacian $\Ll_N$ defined by $\Ll_N = \Dd^{-1/2} \Ll_S \Dd^{-1/2}$ for which all diagonal entries are equal to $1$. A more profound introduction to combinatorial graph theory and the properties of different graph Laplacian can be found in \cite{Chung,GodsilRoyle2001}.  

In graph signal processing, the main research objects are graph signals. A graph signal $x: V \rightarrow \mathbb{R}$ is a function on the vertices $V$ of $G$. We denote the $n$-dimensional vector space of all graph signals by $\mathcal{L}(G)$. As the vertices in $V$ are ordered, we can represent every signal $x$ also as a vector $x = (x(\node{v}_1), \ldots, x(\node{v}_n))^{\intercal}\in \mathbb{R}^n$. On the space $\mathcal{L}(G)$, we have a natural inner product of the form 
$$y^\intercal x := \sum_{i=1}^n x(\node{v}_i) y(\node{v}_i).$$ 
The system $\{\delta_{\node{v}_1}, \ldots, \delta_{\node{v}_n}\}$ of unit vectors forms a canonical orthonormal basis of $\mathcal{L}(G)$, where the unit vectors $\delta_{\node{v}_{i'}}$ are defined as $\delta_{\node{v}_{i'}}(\node{v}_i) = \delta_{\subind{i}{i'}}$ for $i,i' \in \{1, \ldots,n\}$. In this article, we will also consider functions on the set of edges $E$ of the graph $G$. We denote the corresponding space of functions by $\mathcal{L}(E)$. 

\subsection{The graph Fourier transform}
An important tool for denoising, filtering and decomposition of graph signals is the graph Fourier transform. On graphs, this transform is defined via the eigendecomposition of the graph Laplacian $\Ll$:
\begin{equation*}
\mathbf{L}=\mathbf{U}\mathbf{M}_{\lambda} \mathbf{U^\intercal}.
\end{equation*}
Here, $\mathbf{M}_{\lambda} = \mathrm{diag}(\lambda) = \text{diag}(\lambda_1,\ldots,\lambda_{n})$ 
is the diagonal matrix containing the increasingly ordered eigenvalues $\lambda_i$, $i \in \{1, \ldots, n\}$, of $\mathbf{L}$ as diagonal entries.
The column vectors $\hat{G} = \{ u_1, \ldots, u_{n}\}$ of the orthonormal matrix $\mathbf{U}$ form an orthonormal basis of eigenvectors for the space $\mathcal{L}(G)$ with respect to the eigenvalues $\lambda_1, \ldots, \lambda_n$. We call $\hat{G}$ the spectrum or the Fourier basis of the graph $G$. 
Given the spectrum $\hat{G}$, the graph Fourier transform of a graph signal $x$, and the inverse Fourier transform are given as 
\begin{equation*}
\hat{x} := \mathbf{U^\intercal}x  = (u_1^\intercal x, \ldots, u_n^\intercal x)^{\intercal}, \quad \text{and} \quad x = \mathbf{U}\hat{x},
\end{equation*}
respectively. The entries $\hat{x}_k = u_k^\intercal x$, $k \in \{1, \ldots, n\}$, of the Fourier transform $\hat{x}$ describe the frequency components of the signal $x$ with respect to the basis functions $u_i$. As in classical settings like the euclidean space, the Fourier transform on graphs is a crucial tool for the processing and manipulation of graph signals. Overviews about possible uses of the graph Fourier transform, in particular as analysis or decomposition tool for graph signals, can, for instance, be found in \cite{erb2019,Ortega2018,Pesenson2008,shuman2016,StankovicDakovicSejdic2019}.   

\subsection{Positive definite GBFs for signal approximation on graphs} 
\label{sec:GBF}
The usage of positive definite GBFs provides a simple and efficient tool for the interpolation and approximation of graph signals if only a few samples $x(\node{w}_i)$ of a graph signal $x$ are known on a fixed sampling set $W = \{\node{w}_1, \ldots, \node{w}_N\} \subset V$, see \cite{erb2019b,erb2020}. The theory follows closely a corresponding theory on scattered data approximation with positive definite radial basis functions (RBFs) in the euclidean space \cite{SchabackWendland2003,We05}.

In this method on graphs, the approximation spaces are built upon generalized shifts of a GBF $f$. These are defined in terms of a convolution operator $\mathbf{C}_y$ acting on a signal $x$ as   
\[\mathbf{C}_y x = \mathbf{U}\mathbf{M}_{\hat{y}}\mathbf{U^\intercal} x. \] 
A function $f \in \mathcal{L}(G)$ is called a positive definite GBF if the matrix 
\[ \mathbf{K}_{f} = \begin{pmatrix} \mathbf{C}_{\delta_{\node{v}_1}} f(\node{v}_1) & \mathbf{C}_{\delta_{\node{v}_2}} f(\node{v}_1) & \ldots & \mathbf{C}_{\delta_{\node{v}_n}} f(\node{v}_1) \\
\mathbf{C}_{\delta_{\node{v}_1}} f(\node{v}_2) & \mathbf{C}_{\delta_{\node{v}_2}} f(\node{v}_2) & \ldots & \mathbf{C}_{\delta_{\node{v}_n}} f(\node{v}_2) \\
\vdots & \vdots & \ddots & \vdots \\
\mathbf{C}_{\delta_{\node{v}_1}} f(\node{v}_n) & \mathbf{C}_{\delta_{\node{v}_2}} f(\node{v}_n) & \ldots & \mathbf{C}_{\delta_{\node{v}_n}} f(\node{v}_n)
\end{pmatrix}\]
is symmetric and positive definite. Here $\{\delta_{\node{v}_1}, \ldots \delta_{\node{v}_n}\}$ denotes the standard unit basis in $\mathcal{L}(G)$. 

The signals $\mathbf{C}_{\delta_{\node{v}_i}} f$ can be interpreted as generalized translates of the basis function $f$ on the graph $G$. In fact, if $G$ has a group structure and the spectrum $\hat{G}$ consists of properly scaled characters of $G$, then the basis functions $\mathbf{C}_{\delta_{\node{v}_i}} f$ are shifts of the signal $f$ by the group element $\node{v}_i$. 

A GBF $f$ is positive definite if and only if $\hat{f}_k > 0$ for all $k \in \{1, \ldots, n\}$ (for the derivations, see \cite{erb2019b}). Further, the kernel $K_f(\node{v},\node{w}) := \mathbf{C}_{\delta_{\node{w}}}f(\node{v})$ has the Mercer decomposition
\[ K_f(\node{v},\node{w}) = \mathbf{C}_{\delta_{\node{w}}}f(\node{v}) = \sum_{k=1}^n \hat{f}_k \, u_k(\node{v}) \, u_k(\node{w}).\]
In this way, a positive definite GBF $f$ induces, in a natural way, an inner product $\langle x,y \rangle_{K_f}$ and a norm $\| x \|_{K_f}$ as 
\[ \langle x , y \rangle_{K_f} = 
\sum_{k=1}^n \frac{\hat{x}_k \, \hat{y}_k}{\hat{f}_k} = \hat{y}^\intercal \mathbf{M}_{1/\hat{f}} \, \hat{x} \quad \text{and} \quad \| x \|_{K_f} = \sqrt{\sum_{k=1}^n \frac{\hat{x}_k^2}{\hat{f}_k}}. \]
The space $\mathcal{L}(G)$ of signals endowed with this inner product is a reproducing kernel Hilbert space $\mathcal{N}_{K_f}$ with reproducing kernel $K_f$ (a systematic study of such Hilbert spaces is given in \cite{Aronszajn1950}).

To obtain a GBF approximation $x_*$ of a signal $x$ for a few known samples $x(\node{w}_i)$, $i \in \{1, \ldots, N\}$, we consider the solution of the following regularized least squares (RLS) problem (more details are given in \cite{erb2020,Rifkin2003,SmolaKondor2003})
\begin{equation} \label{eq:RLSfunctional}
x_* = \underset{y \in \mathcal{N}_{K_f}}{\mathrm{argmin}} \left( \frac{1}{N} \sum_{i=1}^N |x(\node{w}_i)-y(\node{w}_i)|^2 + \gamma \|y\|_{K_f}^2 \right), \quad \gamma > 0.
\end{equation}

The first term in this RLS functional is referred to as data fidelity term that ensures that the values $x_*(\node{w}_i)$ are close to the values $x(\node{w}_i) \in \Rr$, on the sampling set $W = \{\node{w}_1, \ldots, \node{w}_N\} \subset V$. The second term with the regularization parameter $\gamma >0$ is called regularization term. It forces the optimal solution $x_*$ to have a small Hilbert space norm $\|x_*\|_{K}$.
The representer theorem \cite[Theorem 4.2]{Schoelkopf2002} implies that the minimizer $x_*$ of the RLS functional \eqref{eq:RLSfunctional}
can be uniquely written as 
\begin{equation} \label{eq:representertheorem}
x_*(\node{v}) = \sum_{i = 1}^N c_i \mathbf{C}_{\delta_{\node{w}_i}}f(\node{v}), 
\end{equation}
and is therefore an element of the subspace
\[\mathcal{N}_{K,W} = \left\{x \in \mathcal{L}(G) \ \Big| \ x = \sum_{i=1}^N c_i \mathbf{C}_{\delta_{\node{w}_i}}f, \; c_i \in \Rr \right\}. \]
Furthermore, the coefficients $(c_1, \ldots, c_N)^\intercal$ in \eqref{eq:representertheorem} can be calculated as the solution of the linear system (cf. \cite{Rifkin2003}, \cite[Theorem 1.3.1.]{Wahba1990}) 
\begin{equation} \label{eq:computationcoefficients} 
 \left(\begin{array}{c} \phantom{\node{c}_1} \\ \phantom{\vdots} \\ \phantom{c_N} \end{array}\right. \hspace{-0.6cm} \underbrace{ \begin{pmatrix} \mathbf{C}_{\delta_{\node{w}_1}}f(\node{w}_1) & \ldots & \mathbf{C}_{\delta_{\node{w}_N}}f(\node{w}_1) \\
\vdots & \ddots & \vdots \\
\mathbf{C}_{\delta_{\node{w}_N}}f(\node{w}_1) & \ldots & \mathbf{C}_{\delta_{\node{w}_N}}f(\node{w}_N)
\end{pmatrix}}_{\mathbf{K}_{f,W}} + \gamma N \mathbf{I}_N \hspace{-0.75cm} \left. \begin{array}{c} \phantom{\node{c}_1} \\ \phantom{\vdots} \\ \phantom{c_N} \end{array}\right) \begin{pmatrix} \node{c}_1 \\ \vdots \\ c_N \end{pmatrix}
= \begin{pmatrix} x(\node{w}_1) \\ \vdots \\ x(\node{w}_N) \end{pmatrix}.
\end{equation}
The principal submatrix $\mathbf{K}_{f,W}$ of the positive definite matrix $\mathbf{K}_f$ is positive definite by the inclusion principle \cite[Theorem 4.3.15]{HornJohnson1985}. The linear system \eqref{eq:computationcoefficients} has therefore a unique solution and provides a unique GBF approximation $x_*$.

For a vanishing regularization parameter $\gamma \to 0$, the limit $\itpx = \lim_{\gamma \to 0}x_*$ is uniquely determined by the condition \eqref{eq:representertheorem} and the unique coefficients calculated in \eqref{eq:computationcoefficients} with $\gamma = 0$. 
The resulting signal $\itpx$ interpolates the data $(\node{w}_i,x(\node{w}_i))$, i.e. we have $\itpx(\node{w}_i) = x(\node{w}_i)$ for all $i \in \{1, \ldots, N\}$. 

\begin{exa} \label{exa:GBFs}
Two main examples of positive definite GBFs are the following: 
\begin{enumerate}
\item[(i)] The diffusion kernel on a graph \cite{KondorLafferty2002} is defined by the Mercer decomposition
\[\mathbf{K}_{f_{e^{-t \mathbf{L}}}} = e^{ -t \mathbf{L}} = \sum_{k=1}^n e^{ -t \lambda_k} u_k u_k^{\intercal},\]
where $\lambda_k$, $k \in \{1, \ldots, n\}$, denote the eigenvalues of the graph Laplacian $\Ll$. This kernel is positive definite for all $t \in \Rr$. The corresponding GBF $f_{e^{-t \mathbf{L}}}$ is uniquely determined by the Fourier transform
\[\hat{f}_{e^{-t \mathbf{L}}} = (e^{-t \lambda_1}, \ldots, e^{-t \lambda_n}). \]
\item[(ii)] Variational or polyharmonic splines on graphs are based on the kernel
$$ \mathbf{K}_{f_{(\epsilon \mathbf{I}_n + \mathbf{L})^{-s}}} = (\epsilon \mathbf{I}_n + \mathbf{L})^{-s} = \sum_{k=1}^n \frac{1}{(\epsilon + \lambda_k)^s} u_k u_k^{\intercal}.$$
This kernel is positive definite for $\epsilon > - \lambda_1$ and $s > 0$. Variational splines are studied in \cite{Pesenson2009,Ward2018interpolating} as
interpolants $\itpx$ that minimize the energy functional $ x^{\intercal}(\epsilon \mathbf{I}_n + \mathbf{L})^{s} x$. 
They can be regarded as GBF interpolants based on the GBF with the Fourier transform
\[ \hat{f}_{(\epsilon \mathbf{I}_n + \mathbf{L})^{-s}} = \ts \left(\frac{1}{(\epsilon + \lambda_1)^s}, \ldots, \frac{1}{(\epsilon + \lambda_n)^s}\right). \]
\end{enumerate}
\end{exa}

\section{How to obtain a partition of unity on graphs} \label{sec:PU}

\subsection{Graph partitioning by metric $J$-center clustering}
A simple way to split a connected graph $G$ in $J$ disjoint clusters $\node{C}_j$, $j \in \{1,\ldots, J\}$ is by minimizing an objective functional that measures the distances of the nodes inside the partitions. For this, a common approach is to consider a metric $\mathrm{d}$ on the graph $G$ and $J$ center nodes $\node{q}_j \in \node{C}_j$ as reference nodes for the clusters $\node{C}_j$. In metric $J$-center clustering (in the literature usually a $k$ is used instead of the $J$), the centers $\node{Q}_J = \{\node{q}_1, \ldots, \node{q}_J\}$ are determined such that the objective functional
\[ h(\node{Q}_J) = \max_{i \in \{1,\ldots,n\}} \min_{j \in \{1,\ldots,J\}} \mathrm{d}(\node{q}_j,\node{v}_i)\]
is minimized. We refer the functional $h(\node{Q}_J)$ to as the fill distance of the node set $\node{Q}_J$. It corresponds to the largest possible distance of a node $\node{v} \in V$ to the closest node in $\node{Q}_J$. As a second interpretation, $h(\node{Q}_J)$ can be regarded as the smallest radius such that every node $\node{v} \in V$ is contained in at least one ball of radius $h(\node{Q}_J)$ centered at a node in $\node{Q}_J$. Finding the set $\node{Q}_J^{\mathrm{opt}}$ that minimizes the fill distance $h$ is in general a NP-hard problem \cite{hochbaumshmoys1985}. A feasible alternative to the calculation of the global minimum $\node{Q}_J^{\mathrm{opt}}$ is to use a greedy procedure for the preallocation of the centers \cite{Gonzalez1985}. In view of our application, we also want to guarantee that every cluster $\node{C}_j$ contains at least one sampling node in $W$. To achieve this, we will only consider center nodes $\node{q}_j$ that lie in the node set $W$.  
The corresponding restricted greedy $J$-center clustering is summarized in Algorithm \ref{alg:greedyjcenter}.

\begin{center}
\begin{algorithm}[H] \label{alg:greedyjcenter}
\small

\caption{Restricted greedy $J$-center clustering based on \newline interpolation nodes}

\vspace{1mm}

\KwIn{The interpolation nodes $W = \{\node{w}_1, \ldots, \node{w}_N\}$, a starting center $\node{q}_1 \in W$ and the number $J$ of partitions.  
}

\vspace{1mm}

\For{$j = 2$ to $J$}
    {select a center $\node{q}_j = \ds \underset{\node{w} \in W}{\mathrm{arg max}} \min_{\node{q} \in \{\node{q}_1, \ldots, \node{q}_{j-1}\}} \dist(\node{q},\node{w})$
    farthest away from $Q_{j-1} = \{\node{q}_1, \ldots, \node{q}_{j-1}\}$.}
 
\vspace{1mm}    

\noindent \textbf{Calculate} the $J$ clusters $$\node{C}_j = \{ \node{v} \in V \ : \ \dist(\node{q}_j,\node{v}) = \ds \min_{\node{q} \in \{\node{q}_1, \ldots, \node{q}_{J}\}} \dist(\node{q},\node{v})\}, \quad j \in \{1, \ldots, J\}.$$

\noindent \textbf{Return} centers $\node{Q}_J = \{\node{q}_1, \ldots, \node{q}_{J}\}$ and clusters $\{\node{C}_1, \ldots, \node{C}_{J}\}$.

\end{algorithm}
\end{center}

\begin{rem}
The greedy selection of the node $\node{q}_j$ in Algorithm \ref{alg:greedyjcenter} is in general not unique. In case of non-uniqueness, a simple selection rule to obtain $\node{q}_j$ out of the maximizers is the lowest-index rule. In this rule, we choose $\node{q}_j$ as the admissible node $\node{w}_k \in \{\node{w}_1, \ldots, \node{w}_N\}$ with the lowest index $k$. 
\end{rem}

If the distribution of the interpolation nodes $W$ inside the graph nodes $V$ is fairly reasonable, the restricted greedy method in Algorithm \ref{alg:greedyjcenter} provides a graph clustering which is close to the optimal solution. This is specified in the following result.  

\begin{thm} Algorithm \ref{alg:greedyjcenter} generates a center set $\node{Q}_J$ with $J$ centers and clusters $\{\node{C}_1, \ldots, \node{C}_{J}\}$ such that each cluster $\node{C}_j$ contains at least one interpolation node of $W$ and the fill distance $h(\node{Q}_J)$ is bounded by
\[h(\node{Q}_J) \leq 2 h(\node{Q}_J^{\mathrm{opt}}) + h(W).\]
In particular, if $W$ satisfies $h(W) \leq \epsilon h(\node{Q}_J^{\mathrm{opt}})$ for some $\epsilon > 0$ then the greedy method in Algorithm \ref{alg:greedyjcenter} yields a $(2 + \epsilon)$-approximation to the optimal solution. 
\end{thm}

\begin{proof}
We proceed similar as in the proof of the original result in which no reduction to a subset is performed, see \cite[Theorem 4.3.]{HarPeled2011}. We divide the proof into two disjoint cases: 

We assume first that every cluster $\node{C}_j^{\mathrm{opt}}$, $j \in \{1, \ldots, J\}$, based on the optimal centers $\node{Q}_J^{\mathrm{opt}}$ contains exactly one center from the set $\node{Q}_J$. For an arbitrary node $\node{v}\in V$, let $\node{q}_j^{\mathrm{opt}}$ be the center of the cluster $\node{C}_j^{\mathrm{opt}}$ in which $\node{v}$ is contained and let $\node{q}_j$ be the unique node in $\node{Q}_J$ which is also contained in $\node{C}_j^{\mathrm{opt}}$. Then, the triangle inequality gives
\[\mathrm{d}(\node{v}, \node{q}_j ) \leq \mathrm{d}(\node{v}, \node{q}_j^{\mathrm{opt}} ) + \mathrm{d}(\node{q}_j^{\mathrm{opt}}, \node{q}_j ) = 2 h(\node{Q}_J^{\mathrm{opt}}).\]
In the second case, there exists a cluster $\node{C}_j^{\mathrm{opt}}$ (with center $\node{q}_j^{\mathrm{opt}}$) that contains two distinct elements $\node{q}_j$ and $\node{q}_k$ of the set $\node{Q}_J$. Without loss of generality we assume that $k < j$, and that $\node{q}_j$ is added later to the center set $\node{Q}_J$ in the $j$-th iteration of the greedy algorithm. The greedy selection in Algorithm \ref{alg:greedyjcenter} picks a node $\node{q}_j$ in the interpolation set $W$ that is farthest away from $Q_{j-1}$. Let $\node{v}_j^{\mathrm{opt}} \in V$ be an element that maximizes 
\[ h(\node{Q}_{j-1}) = \max_{\node{v} \in V} \min_{l \in \{1,\ldots,j-1\}}  \mathrm{d}(\node{q}_l,\node{v}).\]
Further, let $\node{w}_j$ be the element in $W$ that is closest to $\node{v}_j^{\mathrm{opt}}$. Then 
\begin{align*}
h(\node{Q}_J) \leq h(\node{Q}_{j-1}) &= \min_{l \in \{1,\ldots,j-1\}} \mathrm{d}(\node{v}_j^{\mathrm{opt}},\node{q}_l) \leq \mathrm{d}(\node{v}_j^{\mathrm{opt}},\node{q}_k) \\
& \leq \mathrm{d}(\node{w}_j, \node{q}_{k}) + \mathrm{d}(\node{v}_j^{\mathrm{opt}}, \node{w}_j ) \\
 &\leq \mathrm{d}(\node{q}_j, \node{q}_{k}) + \mathrm{d}(\node{v}_j^{\mathrm{opt}}, \node{w}_j ) \\
  &\leq \mathrm{d}(\node{q}_j, \node{q}_{j}^{\mathrm{opt}}) + 
  \mathrm{d}(\node{q}_{j}^{\mathrm{opt}},\node{q}_k ) + \mathrm{d}(\node{v}_j^{\mathrm{opt}}, \node{w}_j ) \\
  &\leq 2 h(\node{Q}_J^{\mathrm{opt}}) + h(W).
\end{align*} \qed
\end{proof}

\subsection{Augmentation of graph clusters}
Although it is possible to use the disjoint clusters $\node{C}_j$, $j \in \{1,\ldots,J\}$, directly as subdomains for a partition of unity on the vertex set $V$, this is not recommendable, as the topological information of the graph $G$ gets lost. Instead, we will enlarge the clusters $\node{C}_j$ to overlapping subdomains $V_j$ such that the partition of unity is subsequently generated subordinate to the covering given by the subdomains $V_j$, $j \in \{1, \ldots, J\}$. The relevance of this domain augmentation step will be visible in the numerical experiments performed in Section \ref{sec:numerics}, as well as in the theoretical error estimates of Theorem \ref{thm:globalerrorestimate}. 

The enlargement of the clusters $\node{C}_j$ can be performed in several ways. In the method described in Algorithm \ref{alg:enlargedomains}, we enlarge the clusters $\node{C}_j$ by all vertices of $V$ that have a graph distance to some element of $\node{C}_j$ less or equal to $r \geq 0$. Note that a suitable choice of the augmentation distance $r$ is important for the performance of the PUM. The parameter $r$ must be chosen large enough to guarantee a proper approximation power of the method. On the other hand, it is not desirable to enlarge the clusters $\node{C}_j$ too much, as the overall costs of the computations on the subdomains $V_j$ increase with increasing $r$. We will investigate this trade-off experimentally in Section \ref{sec:numerics}. 

\begin{center}
\begin{algorithm}[H] \label{alg:enlargedomains}

\small

\caption{Augmentation of clusters with neighbors of distance $r$}

\vspace{1mm}

\KwIn{A cluster $\node{C}_j$ and an enlargement distance $r \geq 0$. 
}

\vspace{1mm}

\textbf{Add} all nodes $\node{v} \in V$ with a graph distance less than $r$ to an element of $\node{C}_j$, i.e., 
$$V_j = \node{C}_j \cup \{\node{v} \in V \ | \ \mathrm{d}(\node{v},\node{w}) \leq r, \; \node{w} \in \node{C}_j\}.$$

\noindent \textbf{Return} subdomain $ V_j $.

\end{algorithm}
\end{center}

\begin{cor} The cover $\{ V_j \}_{j=1}^J$ of the node set $V$ generated by Algorithm \ref{alg:greedyjcenter} and Algorithm \ref{alg:enlargedomains} has the following properties: 
\begin{enumerate}[label=(\roman*)]
\item Each subdomain $V_j$ contains a center node $\node{q}_j$ which is also an element of the sampling set $W$.
\item The distance of each node $\node{v} \in V_j$ to the center node $\node{q}_j$ is at most $h(\node{Q}_J) + r$. In particular, the diameter of each subdomain $V_j$ is bounded by $2 h(\node{Q}_J) + 2 r$.
\end{enumerate}
\end{cor}

\begin{figure}[htbp]
	\centering
	\includegraphics[width= 0.24\textwidth]{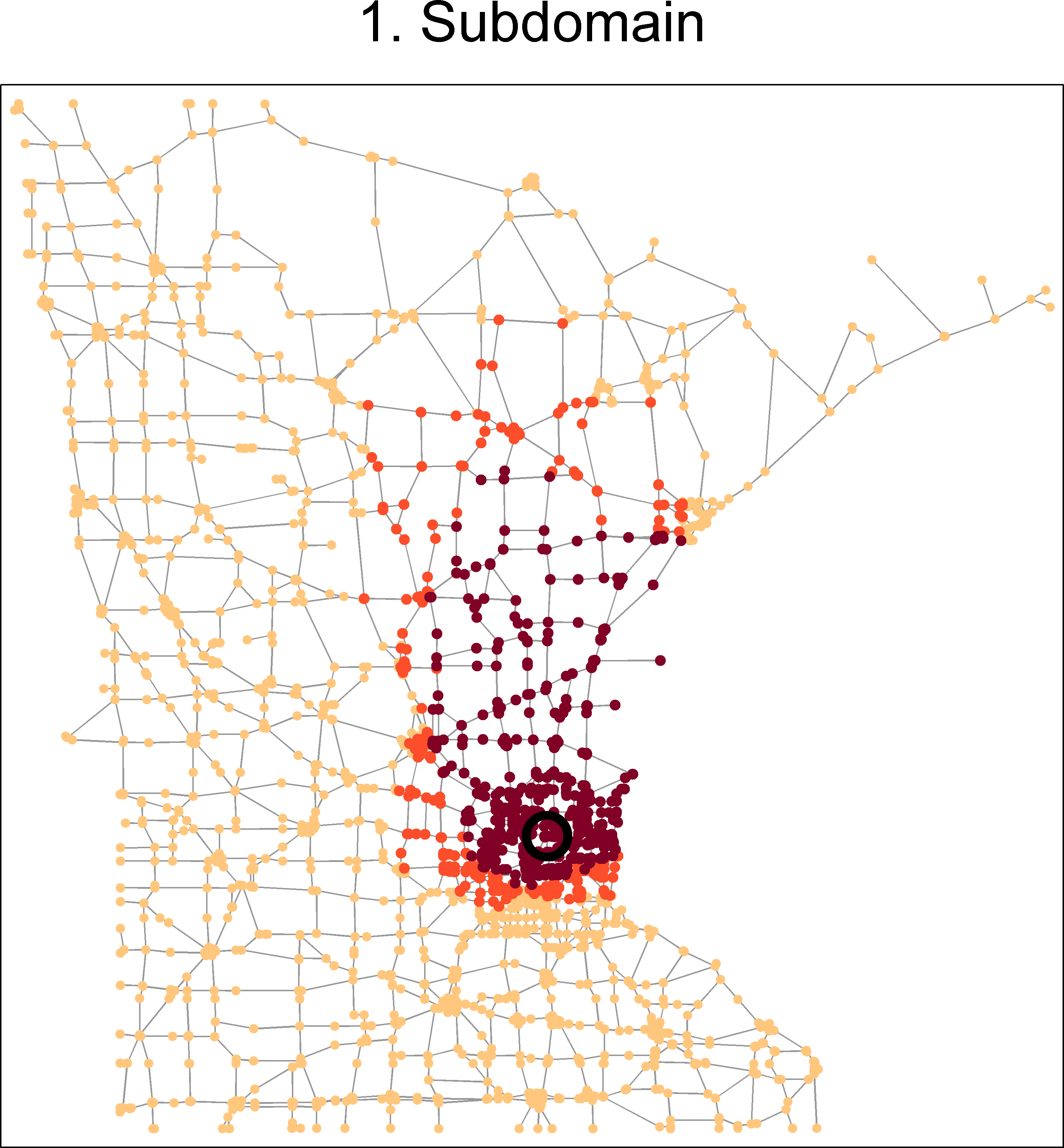}
    \includegraphics[width= 0.24\textwidth]{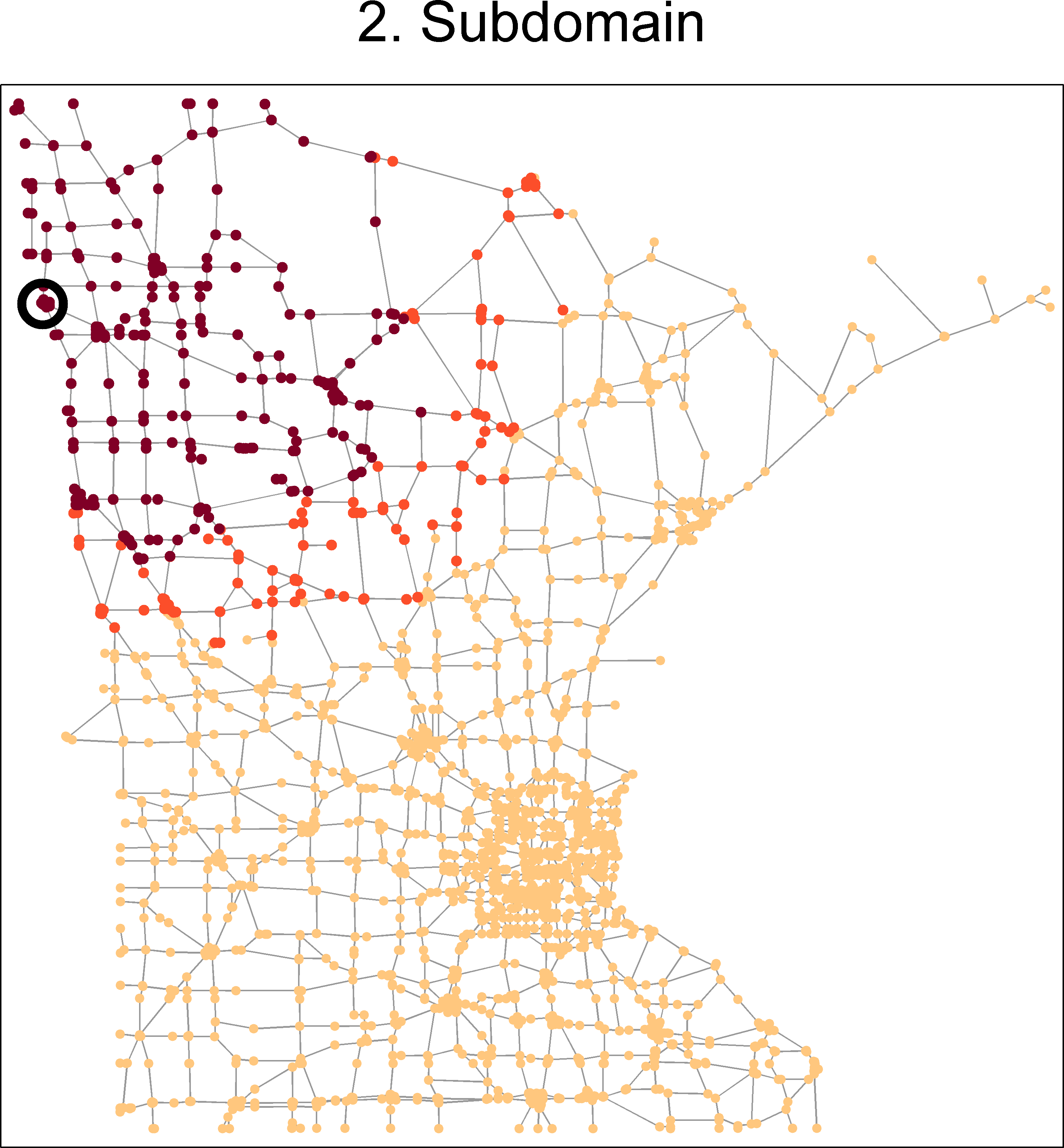}
    \includegraphics[width= 0.24\textwidth]{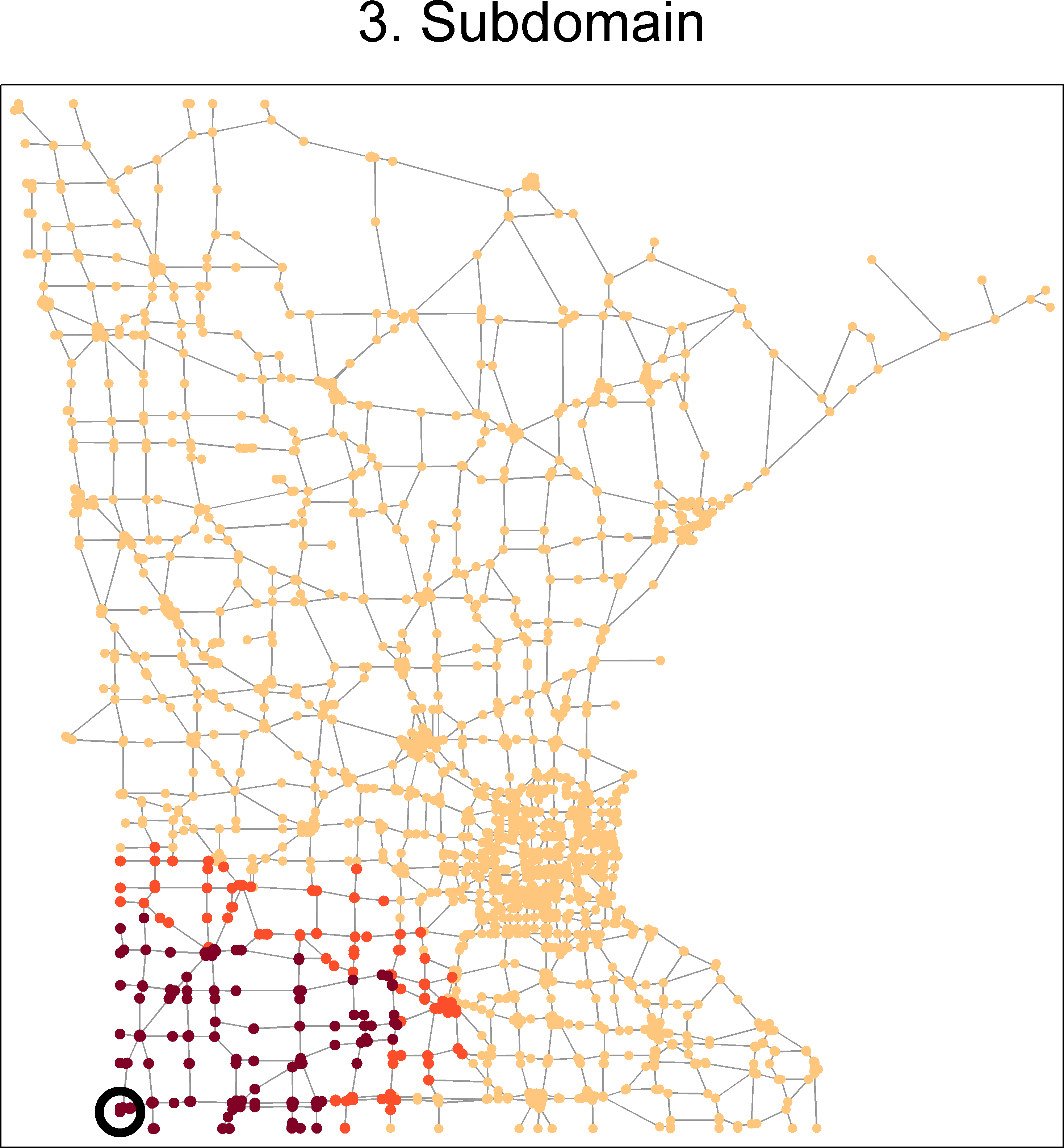}
    \includegraphics[width= 0.24\textwidth]{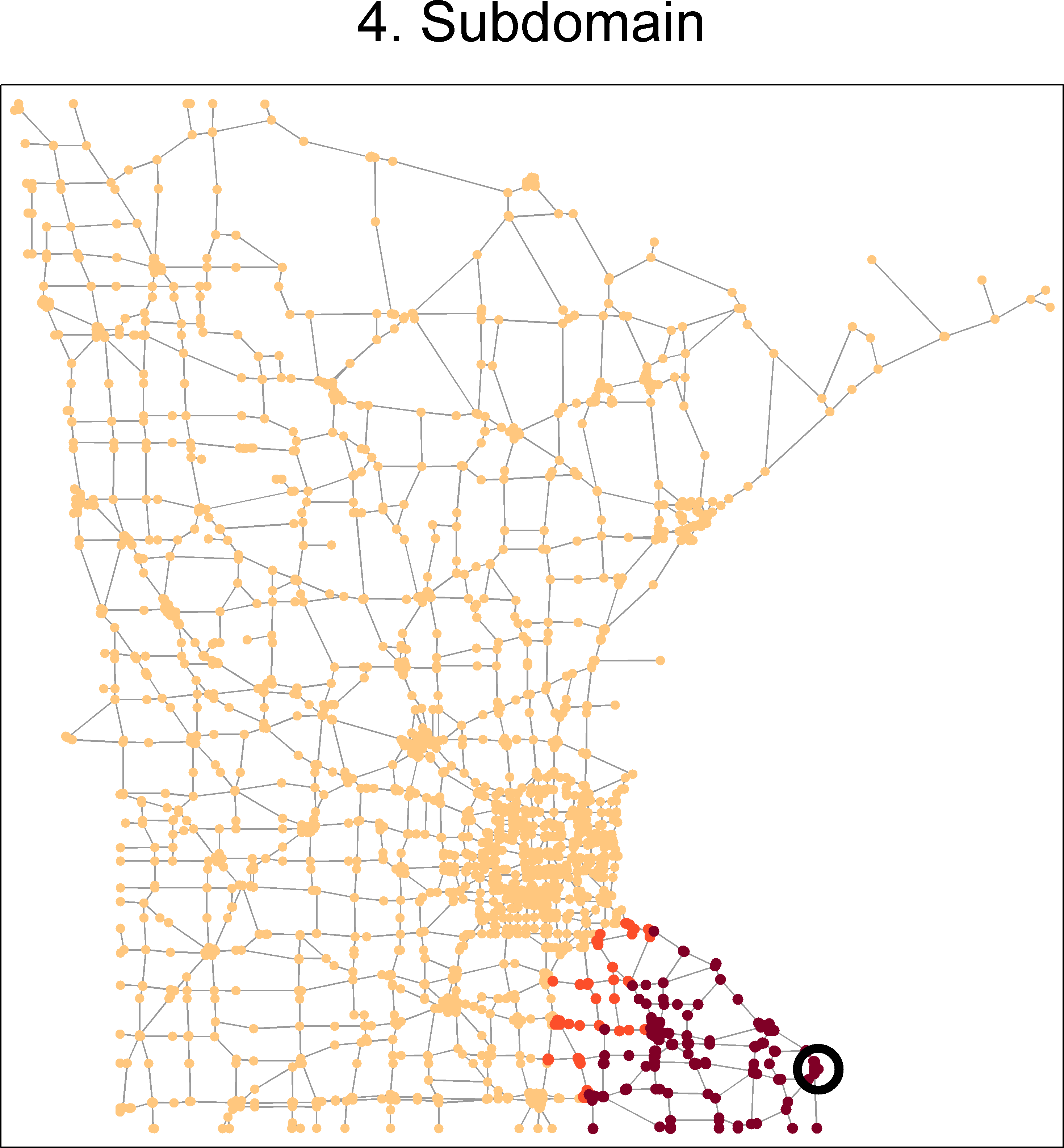} \\[2mm]
    \includegraphics[width= 0.24\textwidth]{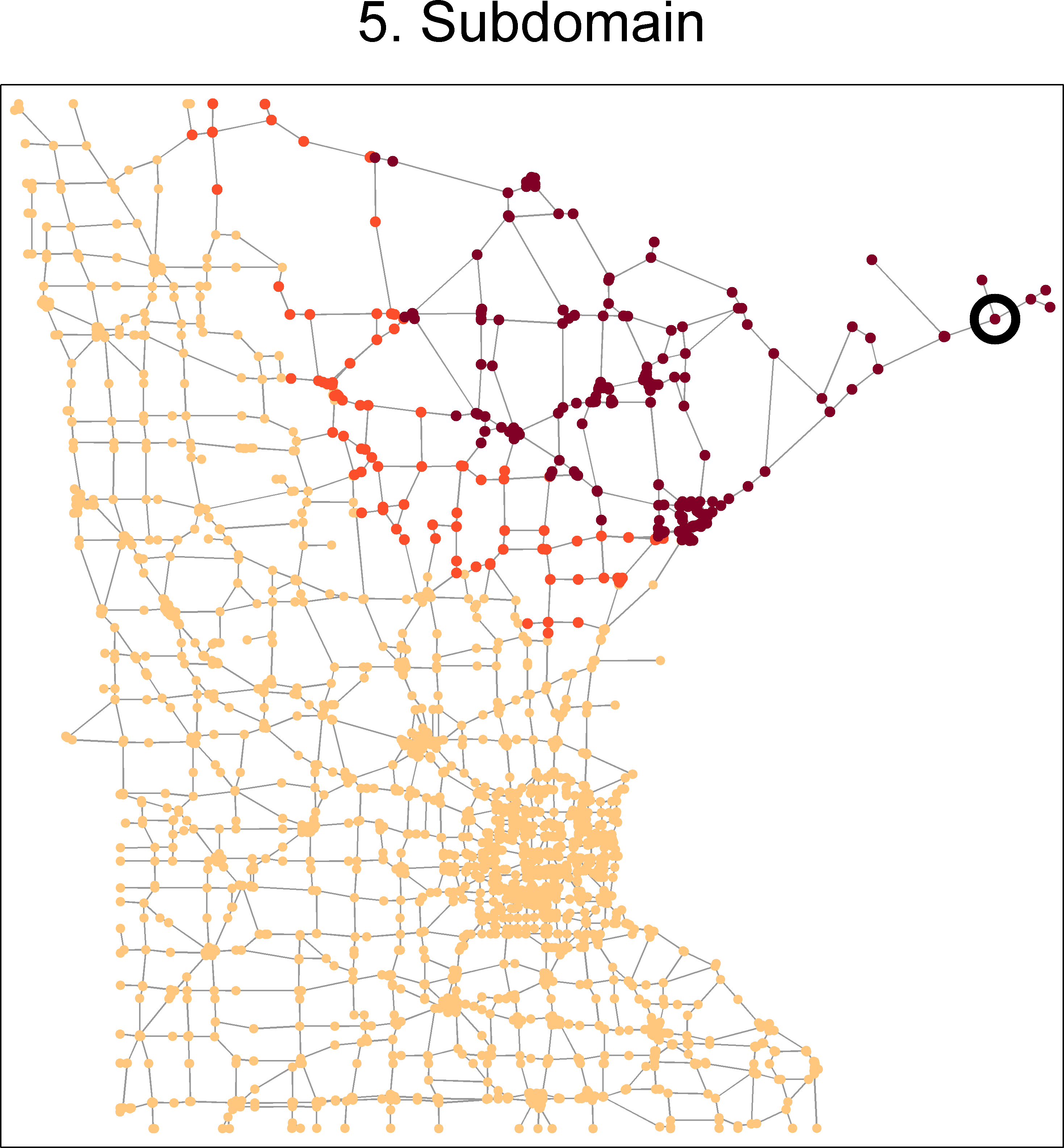}
    \includegraphics[width= 0.24\textwidth]{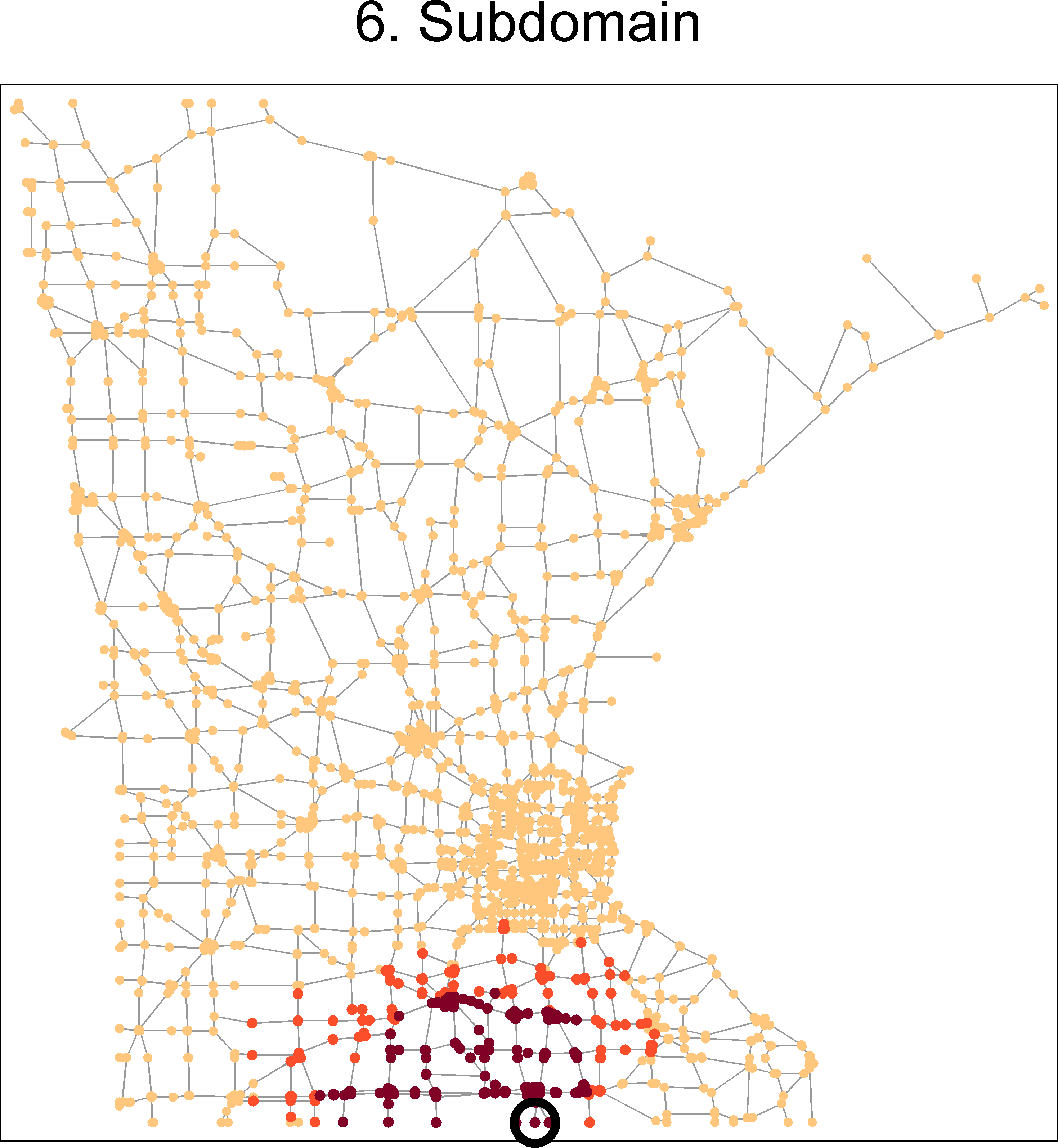}
    \includegraphics[width= 0.24\textwidth]{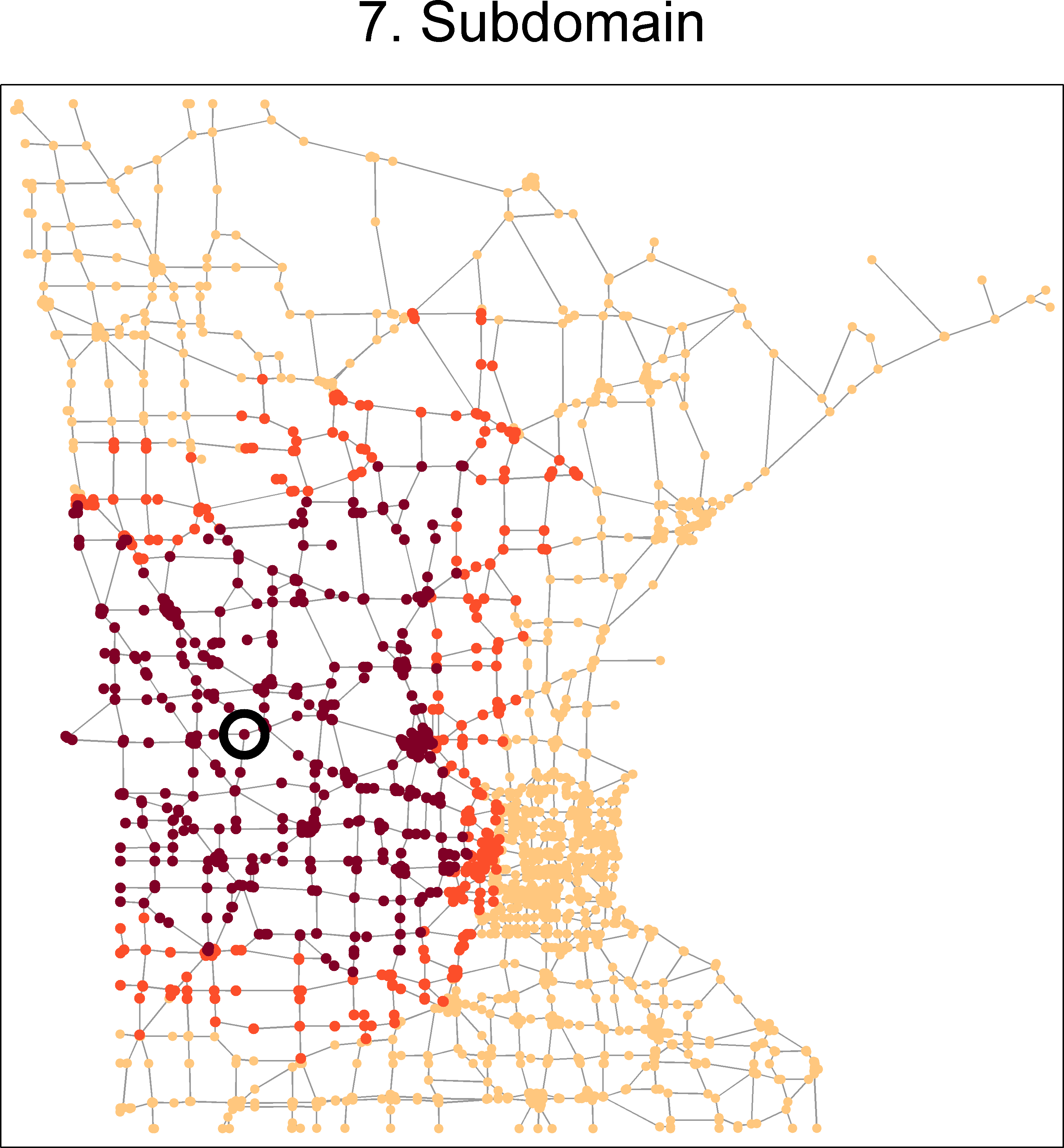}
    \includegraphics[width= 0.24\textwidth]{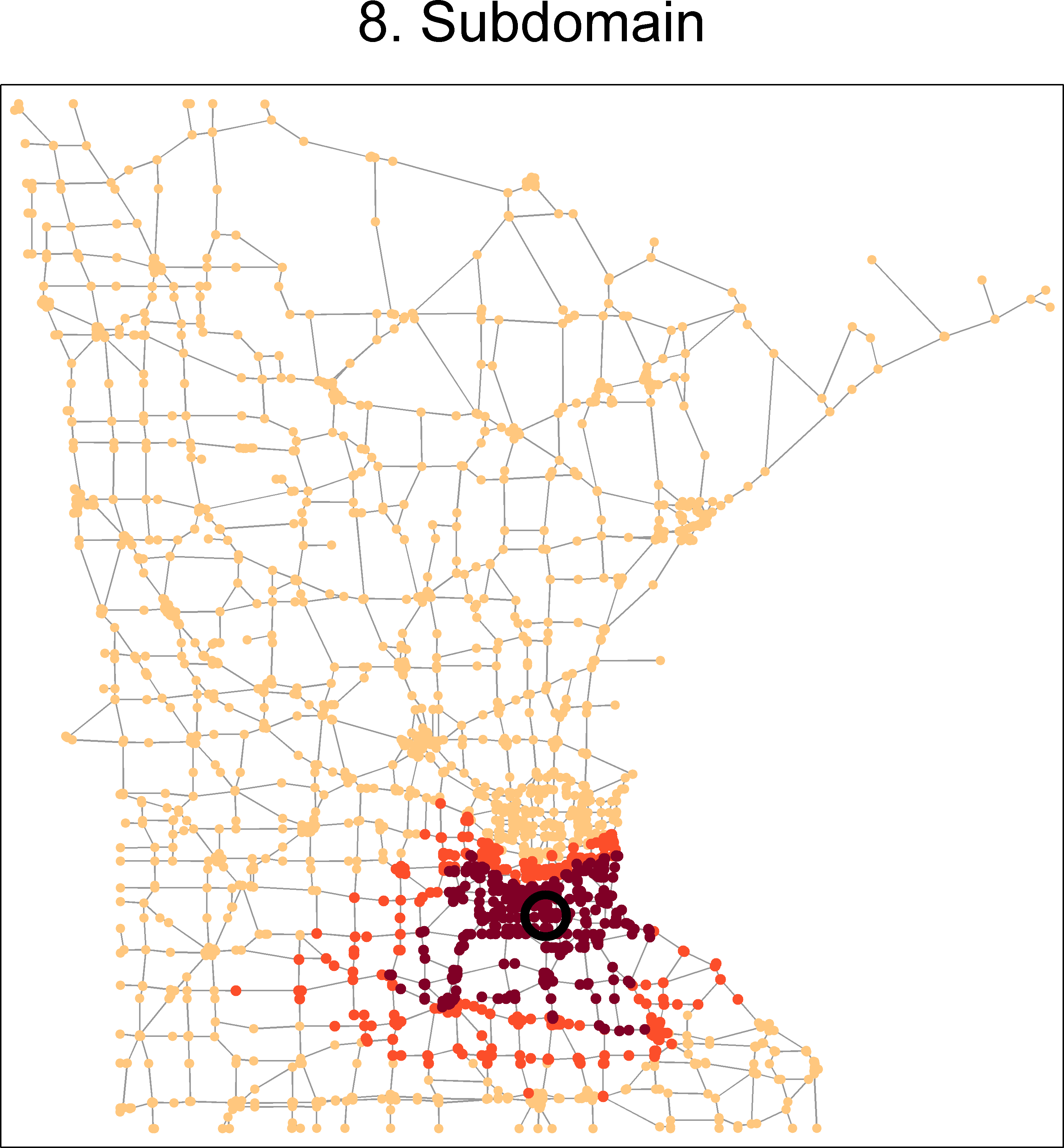} 
	\caption{Partition of the Minnesota graph in $J=8$ subdomains via reduced $J$-center clustering (dark domains) and domain augmentation (red nodes, $r = 8$). The ringed nodes indicate the centers of the subdomains.}
	\label{fig:graphPartitionMinnesota}
\end{figure}

\newpage

\subsection{Generating the partition of unity} \label{sec:genPUM}

Once a cover $\{ V_j \}_{j=1}^J$ of the vertex set $V$ is determined, our next objective is to construct a partition of unity $\{ \varphi^{(j)}\}_{j = 1}^J$ subordinate to this cover, i.e., a set of functions $\varphi^{(j)} \in \mathcal{L}(G)$, $j \in \{1, \ldots, J\}$, such that
\[\mathrm{supp} (\varphi^{(j)}) \subseteq V_j, \quad \varphi^{(j)} \geq 0, \quad \text{and} \quad
\sum_{j = 1}^J \varphi^{(j)}(\node{v}) = 1 \quad \text{for all $\node{v} \in V$}.\]
A simple construction principle to generate such a set of functions is based on Shepard weights \cite{shepard1968}: if $\psi^{(j)}$, $j \in \{1, \ldots, J\}$, denote non-negative weight functions on $V$ with $\mathrm{supp} (\psi^{(j)}) \subseteq V_j$ and $\sum_{j=1}^J \psi^{(j)} > 0$, then the functions
\[ \varphi^{(j)}(\node{v}) = \frac{\psi^{(j)}(\node{v})}{\sum_{j=1}^J \psi^{(j)}(\node{v})}, \quad \node{v} \in V,\]
form a partition of unity subordinate to $\{ V_j \}_{j=1}^J$.

\begin{exa} \label{example:1}
\begin{enumerate}[label=(\roman*)]
\item Choosing $\psi^{(j)}(\node{v}) = \1_{V_j} (\node{v})$, we get the partition of 
unity
\[\varphi^{(j)}(\node{v}) = \frac{\1_{V_j} (\node{v})}{\# \{j \in \{1, \ldots, J\} \,| \, \node{v} \in V_j\}}.\]
\item We consider the disjoint clusters $\node{C}_j$, $j \in \{1, \ldots, J\}$, from Algorithm \ref{alg:greedyjcenter} as subsets of the augmented domains $V_j$. Then, 
choosing $\psi^{(j)}(\node{v}) = \1_{\node{C}_j} (\node{v})$, the corresponding partition of unity is given as
\[\varphi^{(j)}(\node{v}) = \psi^{(j)}(\node{v}) = \1_{\node{C}_j} (\node{v}).\]
\end{enumerate}
\end{exa}

\section{A partition of unity method for GBF approximation on graphs} \label{sec:four}

\begin{center}
\begin{algorithm}[htbp] 

\caption{GBF-PUM approximation on graphs} \label{algorithm1}

\small

\vspace{3mm}

\KwIn{$(i)$ Samples $x(\node{w}_1), \ldots, x(\node{w}_N) \in \Rr$ at the interpolation nodes \newline \phantom{$(i)$} $W = \{\node{w}_1, \ldots, \node{w}_N\}$, \newline
$(ii)$ Number $J$ of subdomains.  
}

\vspace{1mm}

\noindent \textbf{Metric $J$-center clustering on $G$} \newline 
Apply Algorithm \ref{alg:greedyjcenter} to decompose $V$ into $J$ disjoint clusters $\node{C}_j$.

\vspace{1mm}

\noindent \textbf{Create $J$ subdomains $V_j$} \newline
\quad Apply Algorithm \ref{alg:enlargedomains} to augment the clusters $\node{C}_j$ into $J$ overlapping subdomains $V_j = \{\node{v}_{j_1}, \ldots, \node{v}_{j_{n_j}}\}$ with $n_j$ elements.  

\vspace{1mm}

\noindent \textbf{Generate Partition of Unity} \quad
Subordinate to the cover $\{ V_j \}_{j=1}^J$ of $V$ construct a partition of unity $\{ \varphi^{(j)}\}_{j = 1}^J$ such that 

\vspace{-1mm}

\[\mathrm{supp} (\varphi^{(j)}) \subseteq V_j, \; \varphi^{(j)} \geq 0, \; \text{and} \;
\sum_{j = 1}^J \varphi^{(j)}(\node{v}) = 1 \; \text{for all $\node{v} \in V$}.\]

\vspace{-1mm}

\noindent \textbf{Extract local graphs and local Laplacians $\Ll^{(j)}$} \,
For all subdomains $V_j$, $j \in \{1,\ldots, J\}$, generate the subgraphs $G_j = (V_j,E_j,\Ll^{(j)},\mathrm{d})$, with the node sets $V_j$, the edges $E_j = \{(\node{v},\node{w}) \in E \ | \ \node{v},\node{w} \in V_j \}$ and the local Laplacian 

\vspace{-1mm}

\[\Ll^{(j)} = \begin{pmatrix} \Ll_{\subind{j_1}{j_1}} & \cdots & \Ll_{\subind{j_1}{j_{n_j}}} \\
\vdots & \ddots & \vdots \\
\Ll_{\subind{j_{n_j}}{j_1}} & \cdots & \Ll_{\subind{j_{n_j}}{j_{n_j}}} 
           \end{pmatrix}. \]
If required, calculate the spectral decomposition $\Ll^{(j)} = \Uu^{(j)}\Mm_{\lambda^{(j)}} \Uu^{(j) \intercal}$ to define a local graph Fourier transform on the subgraphs $G_j$.

\vspace{1mm}

\noindent \textbf{Construct} the local GBF kernel 
$K_{f^{(j)}}(\node{v},\node{w}) = \Ccc_{\delta_{\node{w}}} f^{(j)}(\node{v})$ on the local graph $G_j$. One possibility is to use the variational spline kernel 

\vspace{-1mm}

\[\Kk_{f^{(j)}} = (\epsilon \mathbf{I}_n + \mathbf{L}^{(j)})^{-s} = \sum_{k=1}^{n_j} \frac{1}{(\epsilon + \lambda_k^{(j)})^s} u_k^{(j)} u_k^{(j)\intercal}, \quad \epsilon > - \lambda_1^{(j)}, \; s > 0.\]

For the sampling sets $W_j = W \cap V_j$ with $N_j \geq 1$ elements, and $\gamma \geq 0$, \textbf{solve} 

\vspace{-1mm}

\begin{equation*} \label{eq:computationcoefficientsGBF} 
 \big( \mathbf{K}_{f^{(j)},W_j} + \gamma N_j \mathbf{I}_{N_j} \big) 
\begin{pmatrix} \node{c}_1^{(j)} \\ \vdots \\ c_{N_j}^{(j)} \end{pmatrix}
= \begin{pmatrix} x(\node{w}_{j_1}) \\ \vdots \\ x(\node{w}_{j_{N_j}}) \end{pmatrix}.
\end{equation*}

\textbf{Calculate} the local GBF approximation $x_*^{(j)}$ on $G_j$:

\vspace{-1mm}

\[ x_*^{(j)}(\node{v}) = \sum_{i=1}^{N_j} c_i^{(j)} \mathbf{C}_{\node{w}_{j_i}} f^{(j)}(\node{v}).\]

\noindent A \textbf{global GBF-PUM approximation} on $G$ is then given as 

\vspace{-1mm}

\[ x_*(\node{v}) = \sum_{j=1}^J \varphi^{(j)}(\node{v}) x_*^{(j)}(\node{v}). \]
\end{algorithm}
\end{center}

In this section, we show how a PUM can be implemented in order to obtain a global approximation on the graph from local GBF approximations on subgraphs $G_j$. We will mainly focus on the construction and the computational steps to obtain the global approximation, the local approximation procedure is already outlined in Section \ref{sec:GBF}. The main computational steps, elaborated formally in 
Algorithm \ref{algorithm1}, are the following:

\begin{enumerate}
\item In a preliminary step, a cover $\{ V_j \}_{j=1}^J$ of the node set $V$ and a partition of unity $\{ \varphi^{(j)}\}_{j = 1}^J$ subordinate to $\{ V_j \}_{j=1}^J$ is constructed. In particular, the following operations are performed:
\begin{itemize}
\item[a)] Algorithm \ref{alg:greedyjcenter} is applied to decompose $V$ into $J$ disjoint clusters $\node{C}_j$.
\item[b)] Algorithm \ref{alg:enlargedomains} is applied to augment the clusters $\node{C}_j$ with neighboring nodes and obtain a cover of overlapping subdomains $V_j$, $j \in \{1, \ldots, J\}$.
\item[c)] As described in Section \ref{sec:genPUM}, Shepard weight functions are introduced on the cover $\{ V_j \}_{j=1}^J$ to obtain a desired partition of unity $\{ \varphi^{(j)}\}_{j = 1}^J$.
\end{itemize}
\item Once the cover $\{ V_j \}_{j=1}^J$ is determined, we extract in a second preliminary step the subgraphs $G_j = (V_j,E_j,\Ll^{(j)},\mathrm{d})$, where $V_j$ is the set of vertices, 
$E_j = \{ e_{\subind{i}{i'}} \in E \ | \ \node{v}_i,\node{v}_{i'} \in V_j \}$ denotes the set of edges connecting two nodes in $V_j$, and the local Laplacian $\Ll^{(j)}$ is the principal submatrix of $\Ll$ corresponding to the nodes of $V_j$. Further, the metric $\mathrm{d}$ on $G_j$ is the induced metric. The $J$-center clustering in Algorithm \ref{alg:greedyjcenter} and the augmentation procedure in Algorithm \ref{alg:enlargedomains} ensure that the subgraphs $G_j$ are connected. 
\item The main computational step consists in the calculation of the local approximations $x_*^{(j)}$ on the subgraphs $G_j$. As outlined in Section \ref{sec:GBF}, this can be done by using a GBF approximation scheme. To this end, on each subgraph a positive definite GBF $f^{(j)}$ has to be chosen. As sampling information for the calculation of the approximation, the samples of the signal $x$ on the subset $W_j = W \cap V_j$ are taken. The modified $J$-center clustering in Algorithm \ref{alg:greedyjcenter} ensures that $W_j \neq \emptyset$. 
\item Once the local approximations $x_*^{(j)}$ are computed, the global GBF-PUM approximation on $G$ is determined by 
\[ x_*(\node{v}) = \sum_{j=1}^J \varphi^{(j)}(\node{v}) x_*^{(j)}(\node{v}). \]
In a similar way, a global interpolating function $x_{\circ}(\node{v})$ can be calculated by using local interpolating functions $x_{\circ}^{(j)}$ instead of the approximants $x_*^{(j)}$. Finally, as described in \cite{erb2020}, the approximant $x_*$ can be processed further by selection functions in order to obtain, for instance, a global classification algorithm for the nodes of the graph $G$.  
\end{enumerate}

\section{From local to global error estimates in PUM approximation} \label{sec:theory}

A partition of unity induces naturally the following approximation spaces:

\begin{dfn}
Let $\{ V_j \}_{j=1}^J$ be a cover of the set $V$ and $\{ \varphi^{(j)}\}_{j = 1}^J$ a subordinate partition of unity. On each of the subgraphs $G_j$, let $\mathcal{N}^{(j)} \subset \mathcal{L}(G_j)$ be a given local approximation space. Then, we define the global partition of unity approximation space $\mathcal{N}^{\mathrm{PUM}}$ as 
\[\mathcal{N}^{\mathrm{PUM}} = \left\{x \in \mathcal{L}(G) \ \Big| \ x = 
\sum_{j=1}^J \varphi^{(j)} x^{(j)}, \; x^{(j)} \in \mathcal{N}^{(j)}   \right\}.\]
\end{dfn}

In the following, our goal is to derive approximation properties of the spaces $\mathcal{N}^{\mathrm{PUM}}$ from given local estimates on the subgraphs $G_j$. As a measure for the smoothness of a function $x \in \mathcal{L}(G)$, we will consider norms of the Laplacian $\Ll$ applied to $x$ or, alternatively, the gradient $\nabla_{\Ll} x$. 
The weighted gradient $\nabla_{\Ll} x$ is defined as a function on the (directed) edges $e_{\subind{i}{i'}} = (\node{v}_i,\node{v}_{i'}) \in E$ of the graph $G$ as
\begin{equation} \label{eq:defgradient}
\nabla_{\Ll} x (e_{\subind{i}{i'}}) = \sqrt{\mathbf{A}_{\subind{i}{i'}}} (x(\node{v}_i) - x(\node{v}_{i'})).
\end{equation}
Based on this definition of the gradient, we can decompose the standard graph Laplacian $\Ll_S$ as
\[x^T \Ll_S x = \sum_{i = 1}^n \sum_{i': \;e_{\subind{i}{i'}}  \in E}
\Aa_{\subind{i}{i'}} \big(x(\node{v}_i) - x(\node{v}_{i'})\big)^2 = (\nabla_{\Ll} x)^T \nabla_{\Ll} x,\]
showing, in particular, that the standard Laplacian $\Ll_S$ is positive semi-definite.
As possible norms for functions $x \in \mathcal{L}(G)$ on the graph $G$ and for functions $z \in \mathcal{L}(E)$ on the edges of $G$, we consider for $1 \leq p \leq \infty$, the usual $\mathcal{L}^p$-norms given by
\begin{align*}
\|x\|_{\normp{G}} &:= \left\{ \begin{array}{ll} \ds \left( \sum_{i = 1}^n |x(\node{v}_i)|^p\right)^{1/p} & \text{for $1 \leq p < \infty$,} \\ 
\ds \max_{i \in \{1, \ldots, n\}} |x(\node{v}_i)| & \text{for $p = \infty$,}\end{array} \right. \\
\|z\|_{\normp{E}} &:= \left\{ \begin{array}{ll} \ds \left( \sum_{i,i': \, e_{\subind{i}{i'}} \in E} |z(e_{\subind{i}{i'}})|^p\right)^{1/p} & \text{for $1 \leq p < \infty$,} \\ 
\ds \max_{i,i': \, e_{\subind{i}{i'}} \in E} |x(e_{\subind{i}{i'}})| & \text{for $p = \infty$.}\end{array} \right.
\end{align*}

In addition, we require for functions $z \in \mathcal{L}(E)$ and $1 \leq p < \infty$ also the hybrid norms $\|z\|_{\norminfp{E}}$ given by
\[\|z\|_{\norminfp{E}} := \max_{i \in \{1, \ldots, n\}} \left(\sum_{i': \, e_{\subind{i}{i'}} \in E } |z(e_{\subind{i}{i'}})|^p\right)^{1/p}.\]
For $p = \infty$, a corresponding definition of this hybrid norm reduces to the already introduced $\mathcal{L}^{\infty}$-norm $\|z\|_{\norminf{E}}$.   

Our main result reads as follows: 
\begin{thm} \label{thm:globalerrorestimate}
Let $\{ V_j \}_{j=1}^J$ be a cover of $V$ and $\{ \varphi^{(j)}\}_{j = 1}^J$ a subordinate partition of unity. For a signal $x \in \mathcal{L}(G)$ we assume that there exist local approximations $x_{*}^{(j)} \in \mathcal{L}(G_j)$ such that, 
for given $1 \leq p \leq \infty$, either the bound $(i)$, the bounds $(i)$, $(ii)$, or the bounds $(i), (ii)$, $(iii)$ are satisfied, where
\[ \begin{array}{lrl}
(i) & \|x - x_{*}^{(j)}\|_{\normp{G_j}} &\leq \error_{p,0}^{(j)}, \\
(ii) & \|\nabla_{\Ll^{(j)}} x - \nabla_{\Ll^{(j)}} x_{*}^{(j)}\|_{\normp{E_j}} &\leq \error_{p,1}^{(j)}, \\
(iii) & \| \Ll^{(j)} x - \Ll^{(j)} x_{*}^{(j)}\|_{\normp{G_j}} &\leq \error_{p,2}^{(j)}.
\end{array}
\]
If $(ii)$ or $(iii)$ are satisfied, we further assume that the partition of unity $\{ \varphi^{(j)}\}_{j = 1}^J$ fulfills the boundary conditions 
\[ (\mathrm{BC}) \quad \nabla_{\Ll} \varphi^{(j)} (e_{\subind{i}{i'}}) = 0 \quad \text{for all edges $e_{\subind{i}{i'}} \in E$ with $\node{v}_i \in V^{(j)}$ and $\node{v}_{i'} \notin V^{(j)}$.}\]
Then, the global error for the approximate signal $x_{*} = 
\sum_{j=1}^J \varphi^{(j)} x_{*}^{(j)} \in \mathcal{N}^{\mathrm{PUM}}$ is bounded by
\[ \begin{array}{lrl}
(i)' & \|x - x_{*}\|_{\normp{G}} &\leq \ds \sum_{j=1}^J \error_{p,0}^{(j)}, \\
(ii)' & \|\nabla_{\Ll} x\! - \!\nabla_{\Ll} x_{*}\|_{\normp{E}} &\leq \ds \sum_{j=1}^J  \left( \error_{p,1}^{(j)} + \|\nabla_{\Ll} \varphi^{(j)} \|_{\norminfp{E}} \,\error_{p,0}^{(j)}\right)\!, \\
(iii)' & \| \Ll x - \Ll x_{*}\|_{\normp{G}} &\leq \ds \sum_{j=1}^J \left(\error_{p,2}^{(j)} + \! \|\nabla_{\Ll} \varphi^{(j)} \|_{\norminfq{E}} \error_{p,1}^{(j)} + \! \|\Ll_S \varphi^{(j)} \|_{\norminf{G}} \error_{p,0}^{(j)} \right)\!,
\end{array} \]
respectively. In the last line, the number $1 \leq q \leq \infty$ is dual to $p$ and determined by $\frac{1}{p} + \frac{1}{q} = 1$.  
\end{thm}

\begin{rem} According to Theorem \ref{thm:globalerrorestimate}, we can expect small global errors if the local errors are small and if the partition of unity is properly constructed. The latter includes that the functions $\varphi^{(j)}$ satisfy the boundary conditions $(\mathrm{BC})$, and that the norms $\|\nabla_{\Ll} \varphi^{(j)} \|_{\norminfp{E}}$ and $\|\Ll_S \varphi^{(j)} \|_{\norminf{G}}$ are small. In analogy to continuous frameworks, this indicates that the functions $\varphi^{(j)}$ of the partition should decay smoothly and vanish at the boundary of the subdomains $V_j$ in order to guarantee small global errors. 

For the boundary condition $(\mathrm{BC})$ to be satisfied, it is necessary that the subdomains $V_j$ of the cover overlap. A partition $\{ \varphi^{(j)}\}_{j = 1}^J$ satisfying $(\mathrm{BC})$ is, for instance, provided in Example \ref{example:1} $(ii)$ as soon as the enlargement distance $r$ in Algorithm \ref{alg:enlargedomains} is chosen large enough. The importance of this vanishing boundary condition and the enlargement parameter $r$ for the global error is further illustrated in the numerical experiments of Section \ref{sec:numerics}. 
\end{rem}

\begin{rem} Obviously, the local errors $\error_{p,0}^{(j)}$, $\error_{p,1}^{(j)}$ and $\error_{p,2}^{(j)}$ depend on the local approximation schemes used for the single subgraphs $G_j$. If, for instance, variational spline kernels are used to approximate signals on $G_j$, the estimates derived in \cite{Pesenson2009} can be used to bound the local errors. For a general GBF interpolation scheme applied to signals on the subgraphs, error bounds are provided in \cite{erb2019b}.
\end{rem}

In order to prove Theorem \ref{thm:globalerrorestimate}, we require a product rule for the graph Laplacian $\Ll$ and the gradient $\nabla_{\Ll}$ when these are applied to the pointwise product $xy$ of two graph signals.

\begin{lem} \label{lem:1} For two signals $x,y \in \mathcal{L}(G)$, we have the following product rules for the gradient $\nabla_{\Ll}$ and the graph Laplacian $\Ll$:
\leavevmode
\begin{enumerate}[label=(\roman*)]
\item \[\nabla_{\Ll}(x y)(e_{\subind{i}{i'}}) = \nabla_{\Ll} x (e_{\subind{i}{i'}}) y(\node{v}_i) + x(\node{v}_{i'}) \nabla_{\Ll} y (e_{\subind{i}{i'}})\]
\item  
\begin{align*} \Ll( x \, y) (\node{v}_i) &= \Ll x (\node{v}_i) y (\node{v}_i) + x (\node{v}_i) \Ll_S y (\node{v}_i) \\
& \quad + \sum_{i': \; e_{\subind{i}{i'}} \in E} \Ll_{\subind{i}{i'}} 
\big(x(\node{v}_i) - x(\node{v}_{i'})\big)\big(y (\node{v}_i) - y (\node{v}_{i'})\big).
\end{align*}
\end{enumerate}
\end{lem}

\begin{proof}
$(i)$ The gradient $\nabla_{\Ll}$ defined in \eqref{eq:defgradient} and applied to the product $x y$ yields
\begin{align*} \nabla_{\Ll}( x \, y) (e_{\subind{i}{i'}}) 
&= (\Aa_{\subind{i}{i'}})^{\frac12} \big( x(\node{v}_i) y(\node{v}_i) - x(\node{v}_{i'}) y(\node{v}_{i'}) \big) \\
&= (\Aa_{\subind{i}{i'}})^{\frac12} \big( x(\node{v}_i) y(\node{v}_i) - x(\node{v}_{i'}\!) y(\node{v}_i) + x(\node{v}_{i'}\!) y(\node{v}_i)- x(\node{v}_{i'}\!) y(\node{v}_{i'}\!) \big) \\
&= \nabla_{\Ll} x (e_{\subind{i}{i'}}) y(\node{v}_i) + x(\node{v}_{i'}) \nabla_{\Ll} y (e_{\subind{i}{i'}}).
\end{align*}
$(ii)$ Similarly, the definition of the graph Laplacian $\Ll$ applied to the product $xy$ of the signal $x$ and $y$ gives
\begin{align} \notag \Ll( x \, y) (\node{v}_i) &= \Ll_{i,i} x(\node{v}_i) y(\node{v}_i) 
+  \sum_{i': \; e_{\subind{i}{i'}} \in E} \Ll_{\subind{i}{i'}} 
x(\node{v}_{i'}) y(\node{v}_{i'}) \\
&= \Ll x(\node{v}_i) y(\node{v}_i)
+ \sum_{i': \; e_{\subind{i}{i'}} \in E} \Ll_{\subind{i}{i'}}  
x(\node{v}_{i'}) \big( y(\node{v}_{i'}) -  y(\node{v}_i) \big). \label{eq:hh123}
\end{align}
As the entries of the standard graph Laplacian $\Ll_S$ are given as
\begin{equation*}    (\Ll_{S})_{\subind{i}{i'}} = 
  \begin{cases}
    \sum_{j=1}^n \Aa_{\subind{i}{j}}, & \text{if } i=i', \\
    \Ll_{\subind{i}{i'}} = - \Aa_{\subind{i}{i'}}, & \text{otherwise},
  \end{cases}
\end{equation*}
we further have the identity
\[ x(\node{v}_i) \Ll_{S} y(\node{v}_i) = \sum_{i': \; e_{\subind{i}{i'}} \in E} \Ll_{\subind{i}{i'}} 
x(\node{v}_i) \big( y(\node{v}_{i'}) -  y(\node{v}_i) \big).\]
Combining this identity with the formula in \eqref{eq:hh123}, we obtain
\begin{align*} \Ll( x \, y) (\node{v}_i) &= \Ll x (\node{v}_i) y (\node{v}_i) + x (\node{v}_i) \Ll_S y (\node{v}_i) \\
& \quad + \sum_{i': \; e_{\subind{i}{i'}} \in E} \Ll_{\subind{i}{i'}}  
\big(x(\node{v}_{i'}) - x(\node{v}_i)\big)\big(y (\node{v}_{i'}) - y (\node{v}_i)\big),
\end{align*}
and, thus, the second statement. \qed
\end{proof}

\begin{lem} \label{lem:2} The product rule in Lemma \ref{lem:1} implies for $1 \leq p,q \leq \infty$, $\frac{1}{p} + \frac{1}{q} = 1$, the following $\mathcal{L}^p$-norm estimates:
\begin{enumerate}[label=(\roman*)]
\item  $\ds\|\nabla_{\Ll}( x \, y)\|_{\normp{E}} \leq \| \nabla_{\Ll} x\|_{\normp{E}} \|y\|_{\norminf{G}} + \|x\|_{\normp{G}}\| \nabla_{\Ll} y\|_{\norminfp{E}}$.\\[-2mm]
\item 
$ \ds \|\Ll( x \, y)\|_{\normp{G}} \! \leq \|\Ll x\|_{\normp{G}} \| y \|_{\norminf{G}} \! + \|x\|_{\normp{G}} \|\Ll_S y\|_{\norminf{G}} \!
+ \|\nabla_{\Ll} x\|_{\normp{E}} \|\nabla_{\Ll} y\|_{\norminfq{E}}$.  
\end{enumerate}
\end{lem}

\begin{proof}
In general lines, the two statements follow from the product formulas in Lemma \ref{lem:1} in combination with H\"olders inequality and the triangle inequality for the $\mathcal{L}^p$-norm. More precisely, we get for $1 \leq p < \infty$ (for $p = \infty$ the argumentation line is analog) the estimates
\begin{align*}
\| \nabla_\Ll( x \, y)\|_{\normp{E}} & = \Big( \sum_{i,i': \, e_{\subind{i}{i'}} \in E} \big| \nabla_{\Ll} x (e_{\subind{i}{i'}}) y(\node{v}_i) + x(\node{v}_{i'}) \nabla_{\Ll} y (e_{\subind{i}{i'}})\big|^p \Big)^{1/p} \\
& \leq \| \nabla_{\Ll} x\|_{\normp{E}} \|y\|_{\norminf{G}} + 
\Big( \sum_{i'=1}^n |x(\node{v}_{i'})|^p \sum_{i: \, e_{\subind{i}{i'}} \in E} \big| \nabla_{\Ll} y (e_{\subind{i}{i'}})\big|^p \Big)^{1/p} \\
&\leq \| \nabla_{\Ll} x\|_{\normp{E}} \|y\|_{\norminf{G}} + \|x\|_{\normp{G}}\| \nabla_{\Ll} y\|_{\norminfp{E}},
\end{align*} 
and
\begin{align*} \| \Ll( x \, y)\|_{\normp{G}} & \leq \| \Ll x \, y \|_{\normp{G}} + \| x \, \Ll_S y \|_{\normp{G}}  \; + \\ 
& \qquad + \left( \sum_{i = 1}^n \Big| \sum_{i' \, e_{\subind{i}{i'}} \in E} \Ll_{\subind{i}{i'}} 
\big(x(\node{v}_{i'}) - x(\node{v}_i)\big)\big(y (\node{v}_{i'}) - y (\node{v}_i)\big) \Big|^p
\right)^{\frac1p} \\
& \leq \| \Ll x \, y \|_{\normp{G}} + \| x \, \Ll_S y \|_{\normp{G}}  \; + \\ 
& \qquad + \left( \sum_{i = 1}^n  \sum_{i': \, e_{\subind{i}{i'}} \in E} 
|\nabla_{\Ll}x (e_{\subind{i}{i'}})|^p \left(\sum_{i': \, e_{\subind{i}{i'}} \in E}  |\nabla_{\Ll}y (e_{\subind{i}{i'}})|^q \right)^{\frac{p}{q}}
\right)^{\frac1p} \\
\\ &\leq \|\Ll x\|_{\normp{G}} \| y \|_{\norminf{G}} + \|x\|_{\normp{G}} \|\Ll_S y\|_{\norminf{G}} 
+ \|\nabla_{\Ll} x\|_{\normp{E}} \|\nabla_{\Ll} y\|_{\norminfq{E}}.
\end{align*} \qed
\end{proof}

\begin{rem}
The identity in Lemma \ref{lem:1} $(ii)$ and the estimates in Lemma \ref{lem:2} are not entirely symmetric in terms of the operator $\Ll$ and the adopted $\mathcal{L}^p$-norms. In Lemma \ref{lem:1} $(ii)$ and Lemma \ref{lem:2} $(ii)$ more symmetry is obtained if the standard Laplacian $\Ll = \Ll_S$ is used. The estimate in Lemma \ref{lem:2} $(i)$ gets symmetric in terms of the norm if the $\mathcal{L}^{\infty}$-norm is used.
\end{rem}

{\noindent \bfseries Proof of Theorem \ref{thm:globalerrorestimate}} \\
$(i) \Rightarrow (i)'$: The approximant $x_* \in \mathcal{N}^{\mathrm{PUM}}$ is a linear combination $x_{*} = \sum_{j=1}^J \varphi^{(j)} x_{*}^{(j)}$ of local approximations, where $(\varphi^{(j)})_{j=1}^{J}$ is a partition of unity subordinate to the cover $\{ V_j \}_{j=1}^J$, satisfying $\mathrm{supp} (\varphi^{(j)}) \subseteq V_j$, $\varphi^{(j)} \geq 0$, and $\sum_{j = 1}^J \varphi^{(j)} = 1$. Therefore, the triangle inequality for the $\mathcal{L}^p$-norm and the local bounds in $(i)$ provide the estimates
\begin{align*}
\|x - x_{*}\|_{\normp{G}} &= \Big\| \sum_{j=1}^J \varphi^{(j)}(x - x_{*}) \Big\|_{\normp{G}} \leq \sum_{j=1}^J \|\varphi^{(j)} (x - x_{*})\|_{\normp{G}} \\ &\leq \sum_{j=1}^J \|\varphi^{(j)}\|_{\norminf{G_j}} \|x - x_{*}\|_{\normp{G_j}} \leq \sum_{j=1}^J \error_{p,0}^{(j)}.
\end{align*}
$(i), (ii) \Rightarrow (i)',(ii)'$: As the functions $\varphi^{(j)}$ satisfy $\mathrm{supp} (\varphi^{(j)}) \subseteq V_j$ and, in addition, the boundary condition $(\mathrm{BC})$, we get
\[\mathrm{supp} (\nabla_{\Ll} \varphi^{(j)}) \subseteq E_j \quad \text{and} \quad\mathrm{supp} (\nabla_{\Ll}(\varphi^{(j)} (x - x_*^{(j)}))) \subseteq E_j.\]
Thus, also
$\nabla_{\Ll} \varphi^{(j)} =  \nabla_{\Ll^{(j)}} \varphi^{(j)}$ and $\nabla_{\Ll}(\varphi^{(j)} (x - x_*^{(j)})) = \nabla_{\Ll^{(j)}}(\varphi^{(j)} (x - x_*^{(j)}))$ on the subset $E_j$. This fact, together with the triangle inequality, gives
\begin{align*}
\|\nabla_{\Ll}x - \nabla_{\Ll} x_{*}\|_{\normp{E}} &\! = \! \Big\| \sum_{j=1}^J \nabla_{\Ll}(\varphi^{(j)}(x - x_{*})) \Big\|_{\normp{E}} \hspace{-0.3cm}\leq \sum_{j=1}^J \|\nabla_{\Ll^{(j)}}\varphi^{(j)} (x - x_{*})\|_{\normp{E_j}}. 
\end{align*}
Now, using Lemma \ref{lem:2} $(i)$ together with the assumptions $(i), (ii)$ of the theorem, we can conclude that
\begin{align*}
\|\nabla_{\Ll}x - \nabla_{\Ll} x_{*}\|_{\normp{E}} & \leq \ds \sum_{j=1}^J  \Big(  \| \varphi^{(j)} \|_{\norminf{G_j}} \|\nabla_{\Ll^{(j)}} (x - x_*^{(j)})  \|_{\normp{E_j}} \\[-3mm] &\hspace{3cm} + \|\nabla_{\Ll^{(j)}} \varphi^{(j)} \|_{\norminfp{E_j}} \| x - x_*^{(j)} \|_{\normp{G_j}} \Big)\\
& \leq \ds \sum_{j=1}^J  \left( \error_{p,1}^{(j)} + \|\nabla_{\Ll} \varphi^{(j)} \|_{\norminfp{E}} \,\error_{p,0}^{(j)}\right).
\end{align*}
$(i), (ii), (iii) \Rightarrow (i)',(ii)',(iii)'$: Since $\mathrm{supp} (\varphi^{(j)}) \subseteq V_j$ and the condition $(\mathrm{BC})$ is satisfied, we get, similarly as before, the inclusions
\[\mathrm{supp} (\Ll \varphi^{(j)}) \subseteq V_j \quad \text{and} \quad\mathrm{supp} (\Ll (\varphi^{(j)} (x - x_*^{(j)}))) \subseteq V_j.\]
This implies that
$\Ll \varphi^{(j)} =  \Ll^{(j)} \varphi^{(j)}$ and $ \Ll(\varphi^{(j)} (x - x_*^{(j)})) = \Ll^{(j)}(\varphi^{(j)} (x - x_*^{(j)}))$ on the subset $V_j$, and, thus
\begin{align*}
\| \Ll x - \Ll x_{*}\|_{\normp{G}} & =  \Big\| \sum_{j=1}^J \Ll (\varphi^{(j)}(x - x_{*})) \Big\|_{\normp{G}} \hspace{-0.2cm}\leq \sum_{j=1}^J \| \Ll^{(j)} \varphi^{(j)} (x - x_{*})\|_{\normp{G_j}}. 
\end{align*}
Now, the statement $(iii)'$ follows from Lemma \ref{lem:2} $(ii)$ and the assumptions $(i)$, $(ii)$, and $(iii)$ of the theorem. \qed

\section{Numerical experiments} \label{sec:numerics}

In this section, we perform a numerical case study in order to illustrate several properties of the PUM on graphs. In particular, we show that the construction of the partition of unity via graph partitioning and domain augmentation, as described in Algorithm \ref{alg:greedyjcenter} and Algorithm \ref{alg:enlargedomains}, is computationally efficient, and that the GBF-PUM approximation scheme in Algorithm \ref{algorithm1} provides accurate global interpolation results. All codes are implemented in a Matlab environment, while the experiments have been carried out on a laptop with an Intel(R) Core i7 6500U CPU 2.50GHz processor and 8.00GB RAM.

The test graph $G$ consists of a road network of the state of Minnesota, is referred to as Minnesota graph and has been retrieved from \cite{RossiAhmed2015}. It consists of $n=2642$ vertices and is characterized by $3304$ edges. As a distance metric $\mathrm{d}$ on $G$, we use the shortest-path distance. As a Laplacian, we use the normalized graph Laplacian $\Ll_N = \Dd^{-1/2} \Ll_S \Dd^{-1/2}$. An exemplary partition of the entire graph $G$ in $J=8$ overlapping subdomains obtained by a combination of Algorithm \ref{alg:greedyjcenter} and Algorithm \ref{alg:enlargedomains} (discussed in Section \ref{sec:PU}) is displayed in Fig. \ref{fig:graphPartitionMinnesota}. 

Regarding the interpolation nodes for the tests, we recursively generate a sequence $W_N$ of sampling sets in $V$ such that $\#W_N=N$, $W_{N-1}$ is contained in $W_N$ and the new node $\node{w}_N \in W_N$ is chosen randomly from $V \backslash W_{N-1}$. As a test function, we use the bandlimited signal $x_B = \sum_{k=1}^{10} u_k$, consisting of the first $10$ eigenvectors of the normalized Laplacian $\Ll_N$ of $G$. Moreover, as a GBF-kernel for interpolation on the subgraphs, we employ the variational spline kernel defined in Example \ref{exa:GBFs} $(ii)$. The chosen parameters for graph interpolation (i.e., $\gamma = 0$) are $\epsilon=0.001$ and $s=2$. 

Once the cover $\{V_j\}_{j=1}^J$ is computed, the partition of unity $\{\varphi_j\}_{j=1}^J$ is generated by the Shepard weights given in Example \ref{example:1} $(ii)$. Using Algorithm \ref{algorithm1} with $J=3$ domains and the augmentation parameter $r = 8$, the local interpolation results for the test function $x_B$ on the $3$ subdomains are shown in Fig. \ref{fig:graphPUMlocalminnesota}. The corresponding composed GBF-PUM interpolant $x_{B,\circ}$ together with an illustration of the absolute error $|x_B - x_{B,\circ}|$ is given in Fig. \ref{fig:graphPUMglobalminnesota}.
 
\begin{figure}[htbp]
	\centering
	\includegraphics[height= 3.8cm]{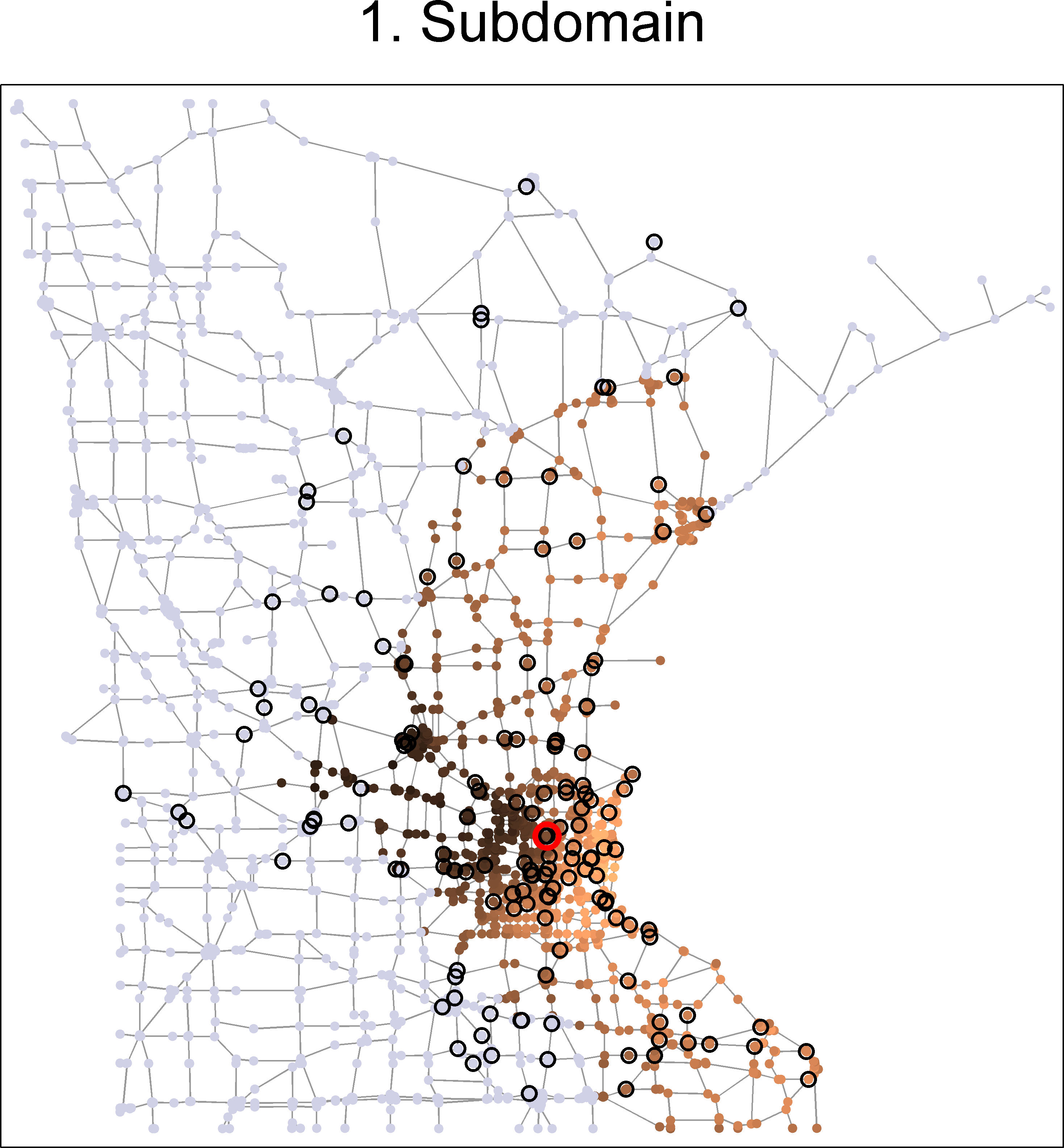}
    \includegraphics[height= 3.8cm]{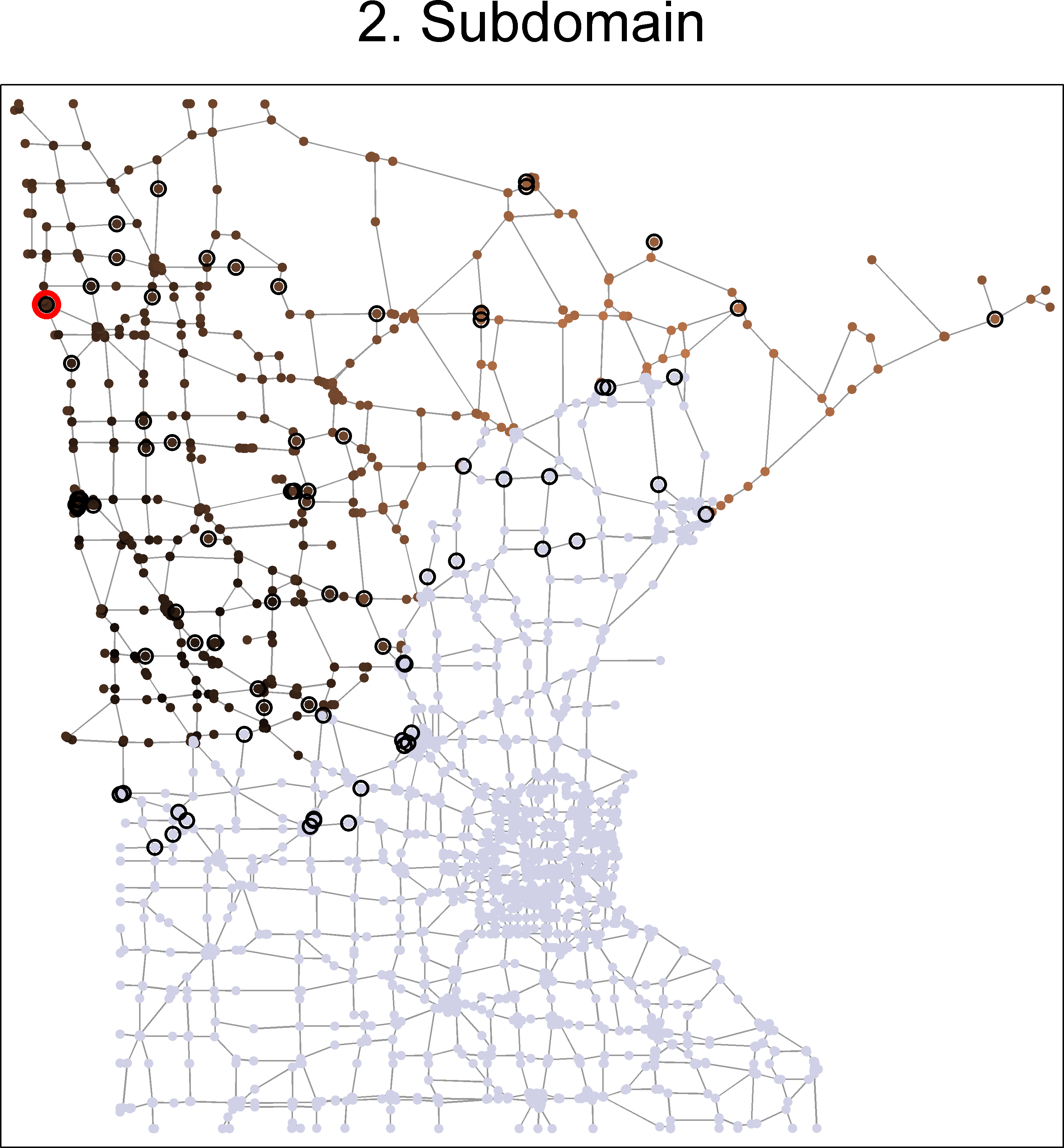}
    \includegraphics[height= 3.8cm]{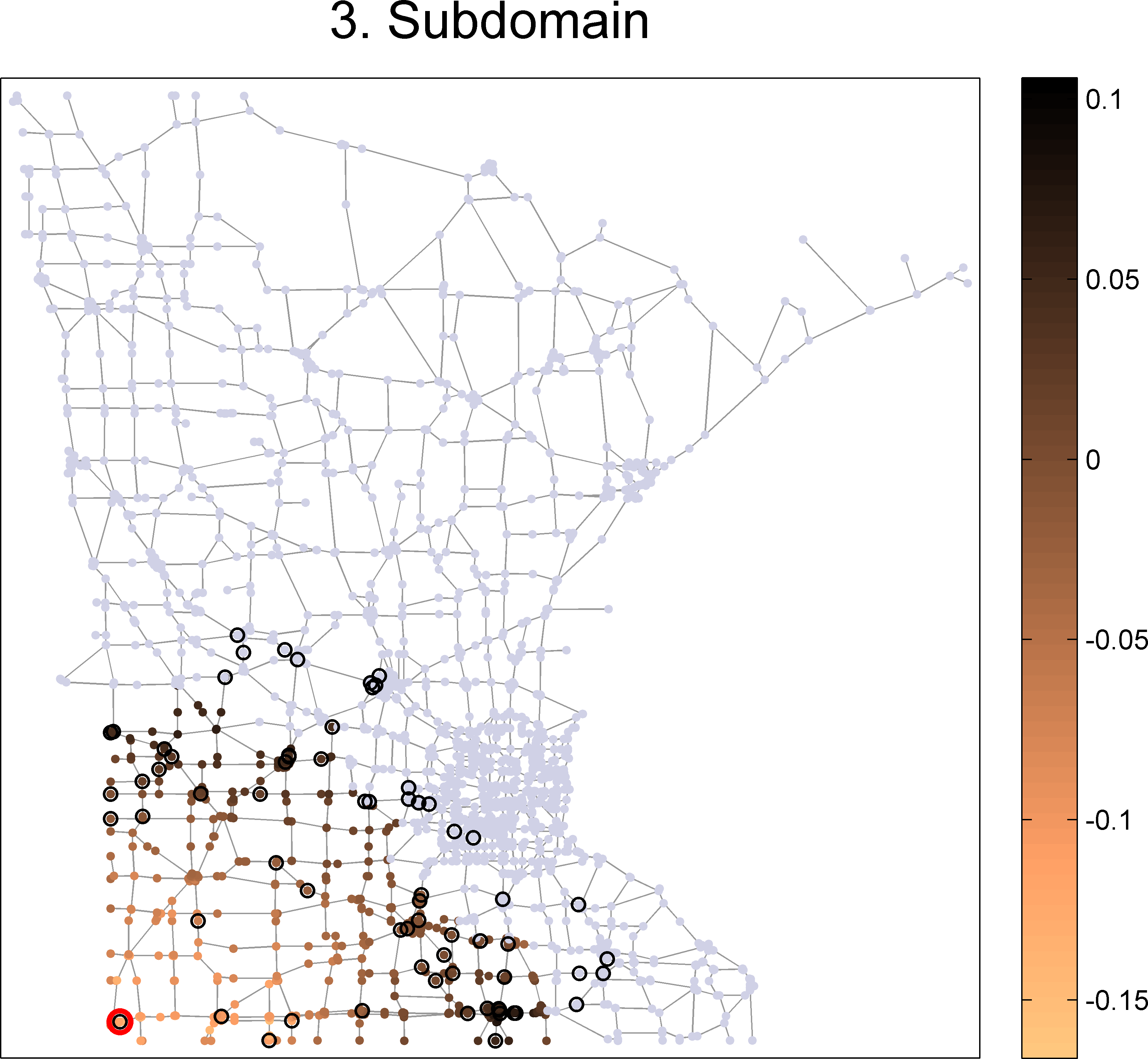}
	\caption{Local GBF-interpolation of the test function $x_B$ on $J = 3$ subdomains of the Minnesota graph based on the respective sampling subsets of $W_{200}$. The black ringed nodes indicate the sampling nodes, while the red ringed nodes denote the centers $\node{q}_j$ of the subdomains.}
	\label{fig:graphPUMlocalminnesota}
\end{figure}

\begin{figure}[htbp]
	\centering
 	\includegraphics[height= 5.1cm]{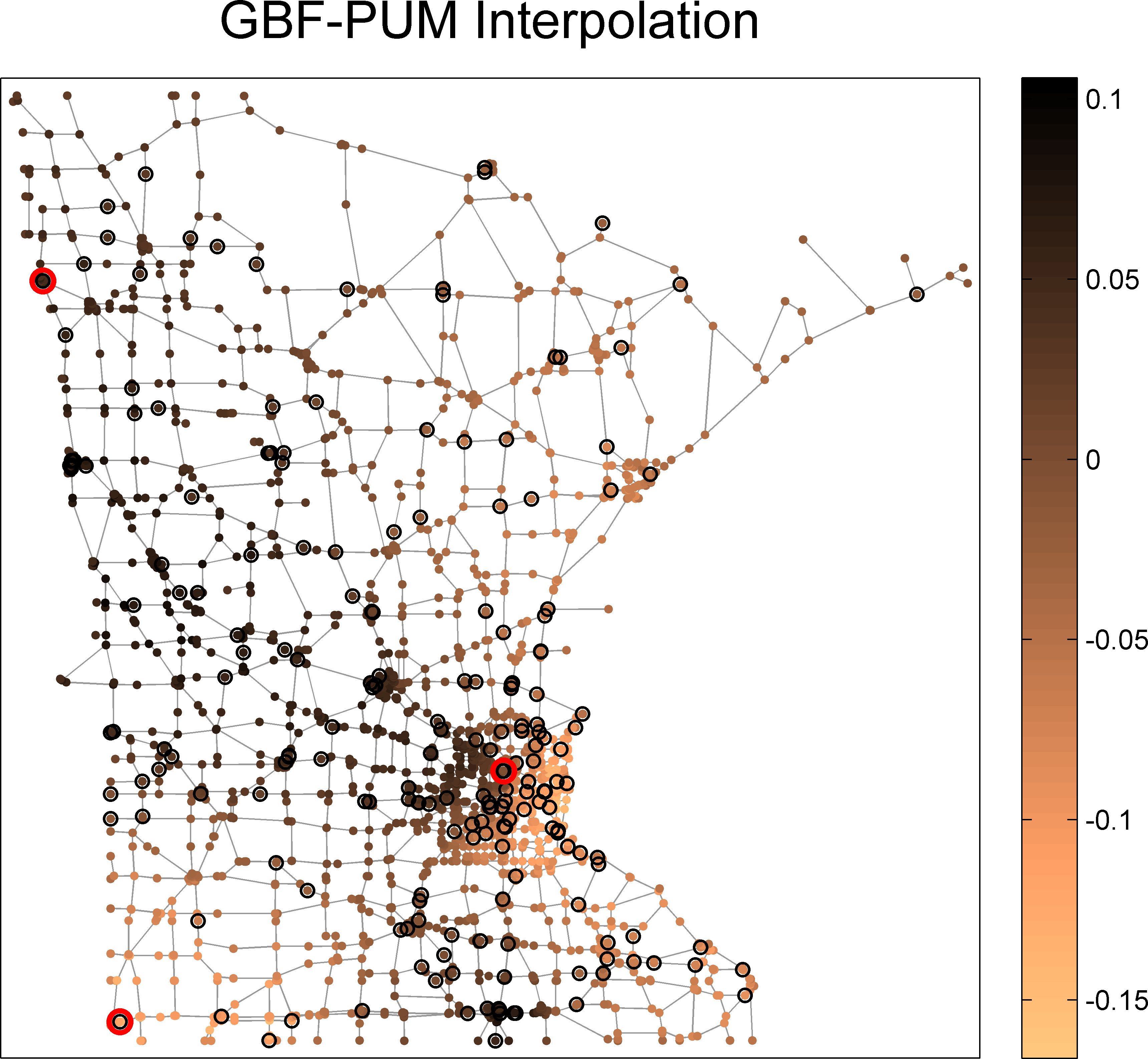} \hspace{1mm}
    \includegraphics[height= 5.1cm]{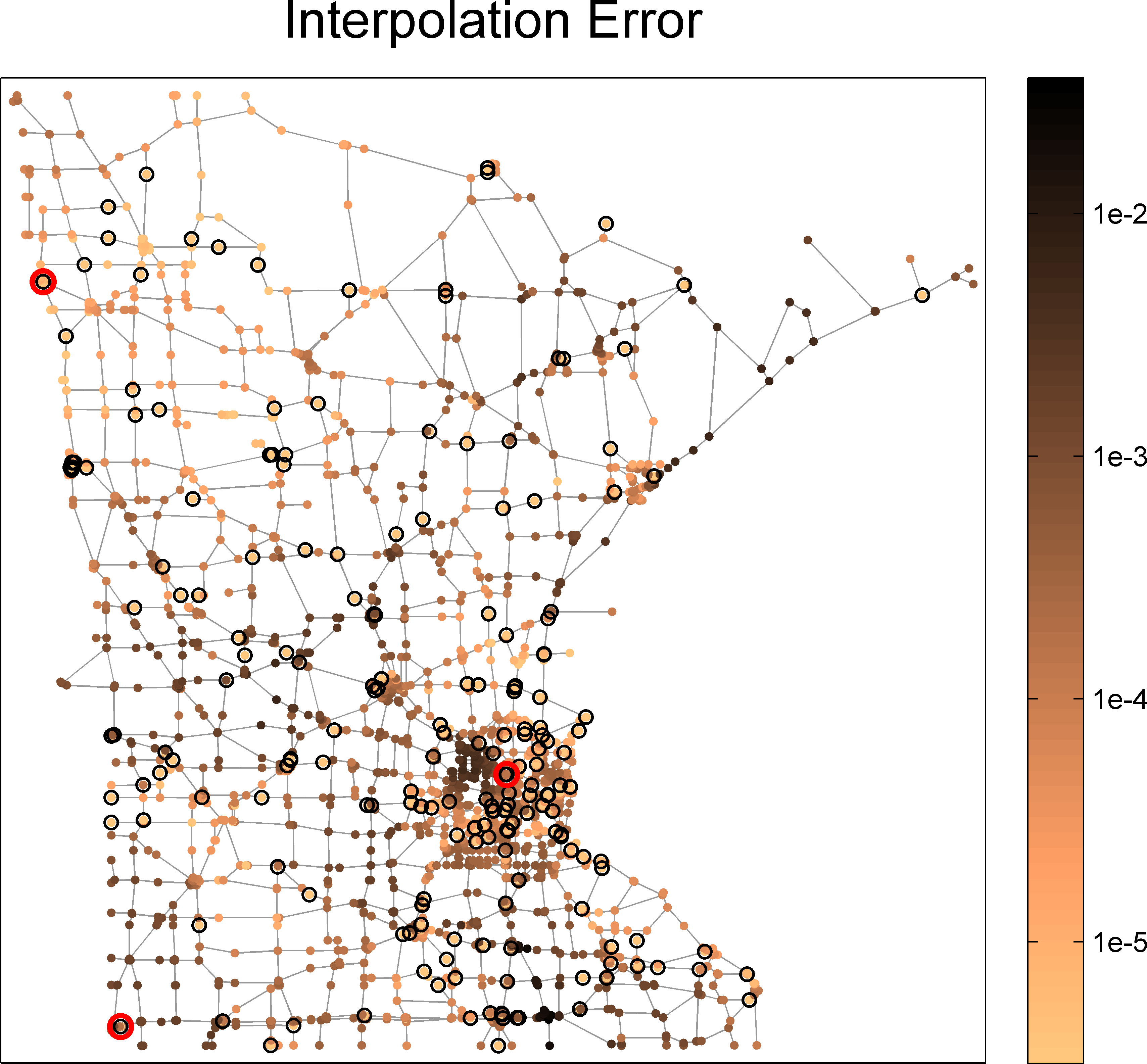} 
	\caption{GBF-PUM method for the Minnesota graph. Left: Global GBF-PUM interpolant $x_{B,\circ}$  of the test function $x_B$ for given samples on the nodes $W_{200}$. Right: Absolute error $|x_B - x_{B,\circ}|$ with respect to the original signal. The ringed nodes indicate the sampling nodes. }
	\label{fig:graphPUMglobalminnesota} 
\end{figure}

To analyze the behavior of the GBF-PUM interpolant in Algorithm \ref{algorithm1} compared to the global GBF method considered in \cite{erb2019b}, we compute the relative root mean square error (RRMSE) and the CPU time expressed in seconds.
In Table \ref{tab_RRMSE_times}, we report RRMSEs and CPU times for several sets of vertices $W_N$. By doing so, we compare the results achieved for PUMs with smaller and larger overlappings of the subdomains in the cover. In particular, Table \ref{tab_RRMSE_times} highlights that, for enlarged subgraph sizes, a general improvement in terms of accuracy is obtained, but also a larger computational cost is necessary. In Fig. \ref{fig1:RRMSE-time} and Fig. \ref{fig2:RRMSE-time}, we represent the same above-mentioned quantities for increasing values of $N$ from $200$ to $2600$. In these two figures, we compare the GBF-PUM interpolation results also with the global GBF method introduced in \cite{erb2019b}. From this study, we can observe that with growing number $N$ of sampling nodes the interpolation errors of the GBF-PUM scheme improve, as predicted by Theorem \ref{thm:globalerrorestimate} (since the local errors decrease). On the other hand, the CPU times show that the use of the PUM for graph signal interpolation is particularly advantageous if the number of interpolation nodes gets large. 

\begin{table}[ht!]
		\begin{center}
			\begin{tabular}{|c|cc|cc|} 
			  \hline
			  & \multicolumn{2}{c|}{smaller overlapping} & \multicolumn{2}{c|}{larger overlapping} \\
			  & \multicolumn{2}{c|}{($r=8$)} & \multicolumn{2}{c|}{($r=12$)} \\
			  \cline{2-5} 
				\rule[-2mm]{0mm}{7mm}
			$N$ & RRMSE & time & RRMSE & time\\
				\hline 
			  \rule[0mm]{0mm}{3ex}
      $132$  & $7.27$e$-2$ &  $0.23$ & $6.16$e${-2}$ & $0.37$ \\
			  \rule[0mm]{0mm}{3ex}
      $264$  & $2.51$e$-2$ &  $0.26$ & $7.95$e${-3}$ & $0.41$ \\
			  \rule[0mm]{0mm}{3ex}
      $528$  & $6.88$e$-3$ &  $0.27$ & $1.69$e${-3}$ & $0.48$ \\
			  \rule[0mm]{0mm}{3ex}
      $1056$ & $8.89$e$-4$ &  $0.35$ & $8.50$e${-5}$ & $0.65$ \\
			  \rule[0mm]{0mm}{3ex}
      $2112$ & $5.55$e$-6$ &  $0.62$ & $5.55$e${-6}$ & $1.06$ \\
				\hline
			\end{tabular}
		\end{center}
    \caption{Interpolation errors and CPU times (in seconds) computed by applying GBF-PUM with $J = 8$ overlapping subdomains and using a variational spline kernel with $\epsilon = 0.001$ and $s=2$.}
    \label{tab_RRMSE_times}
	\end{table}

\begin{figure}[htbp]
	\centering
	\includegraphics[height = 4.5cm]{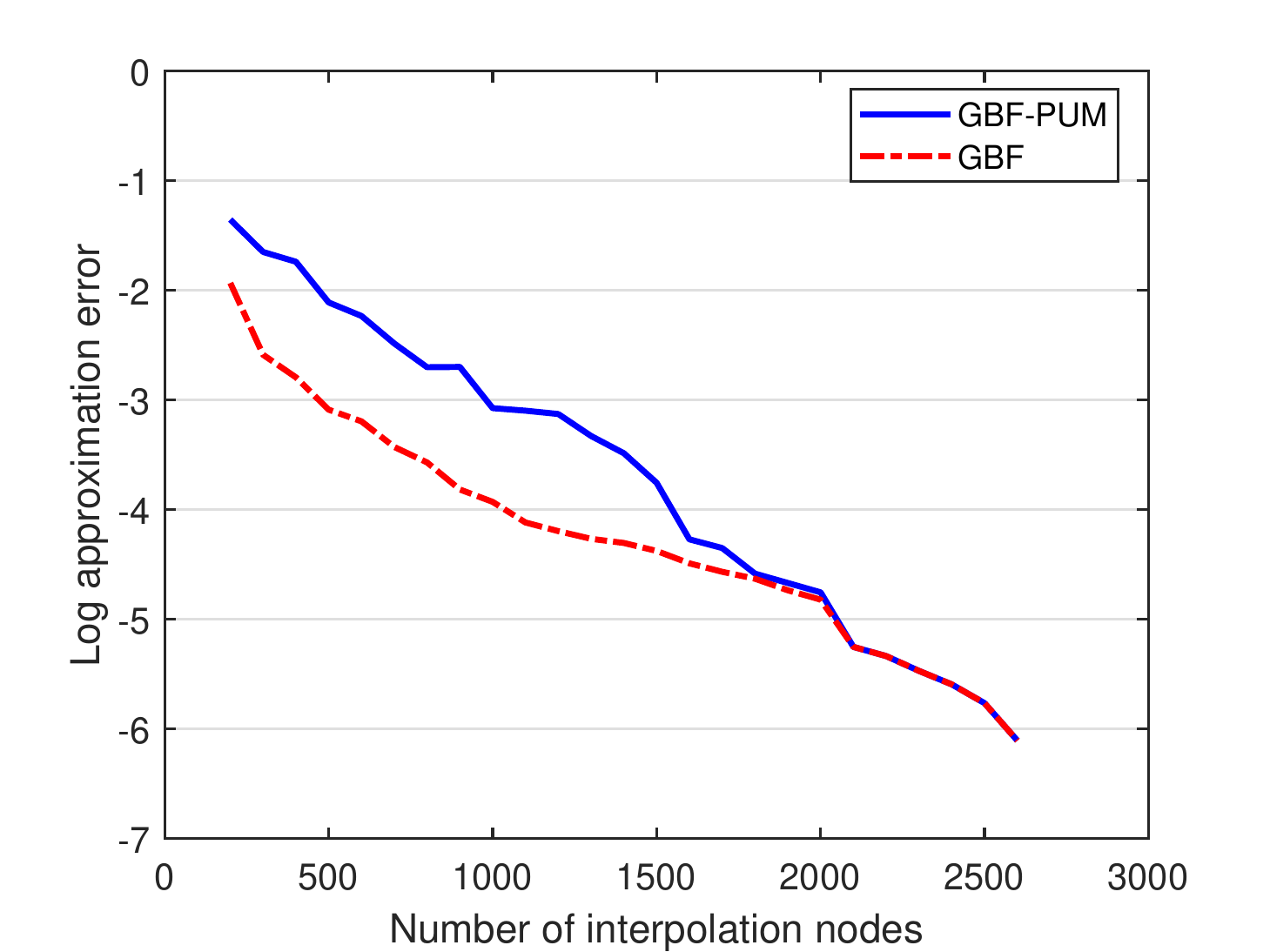} \hskip -0.3cm  
  \includegraphics[height = 4.5cm]{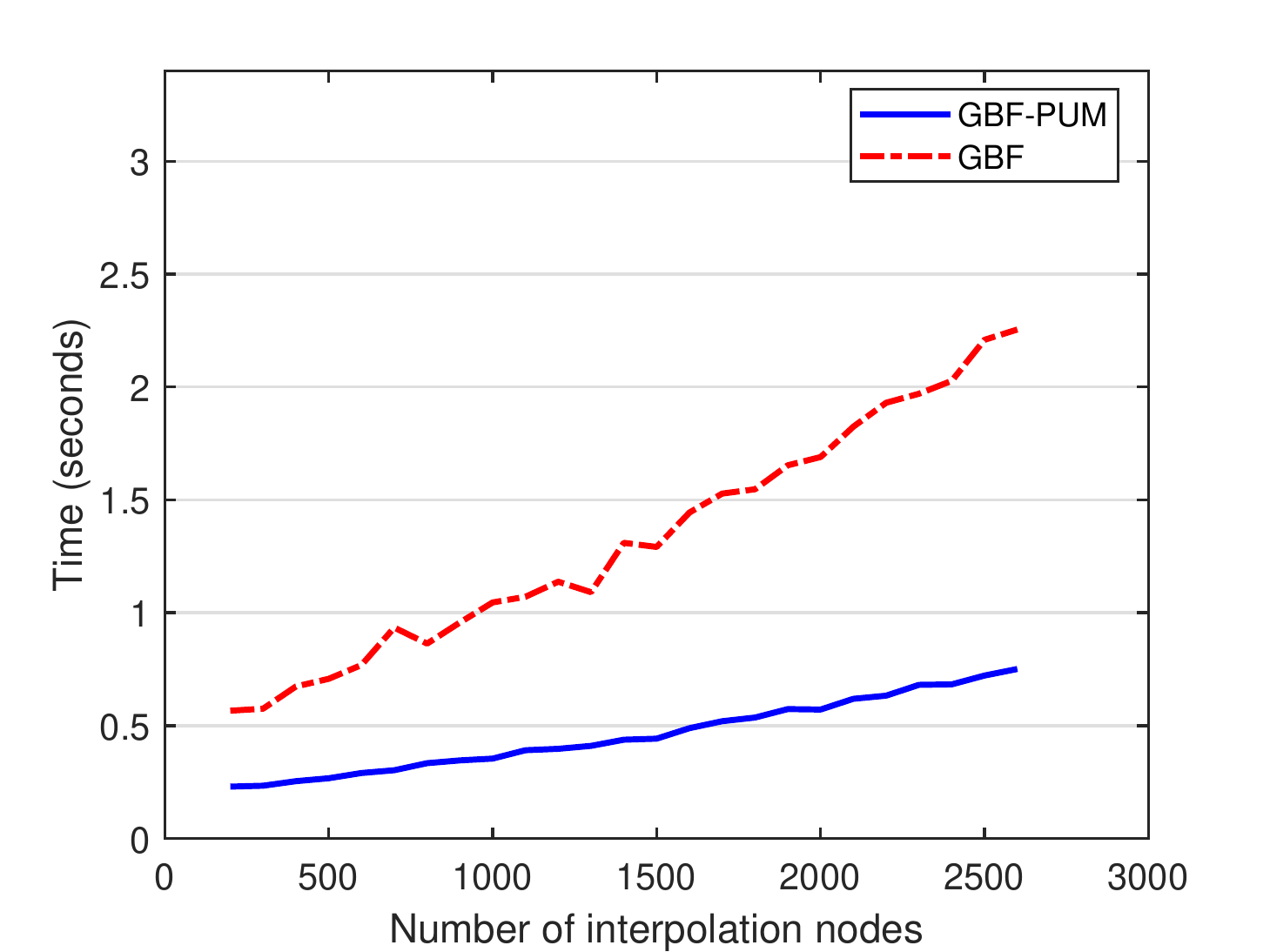}                 
	\caption{RRMSEs and CPU times (in seconds) obtained for the GBF-PUM (with $J = 8$ and $r=8$) and the global GBF scheme in terms of $N$ interpolation nodes. Tests have been carried out by using the variational spline kernel with $\epsilon = 0.001$ and $s=2$.}
	\label{fig1:RRMSE-time}
\end{figure}

\begin{figure}[htbp]
	\centering
	\includegraphics[height = 4.5cm]{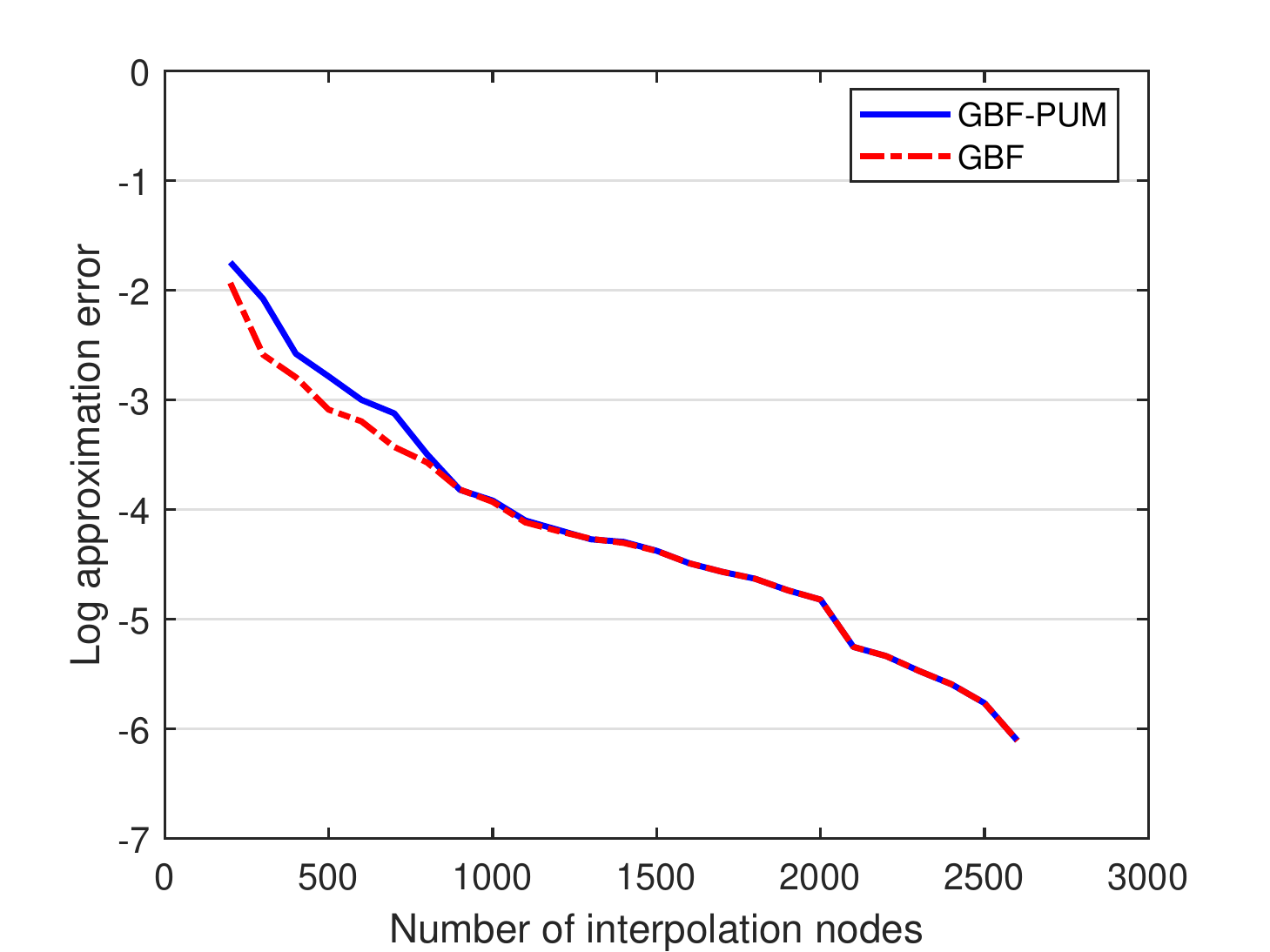} \hskip -0.3cm  
  \includegraphics[height = 4.5cm]{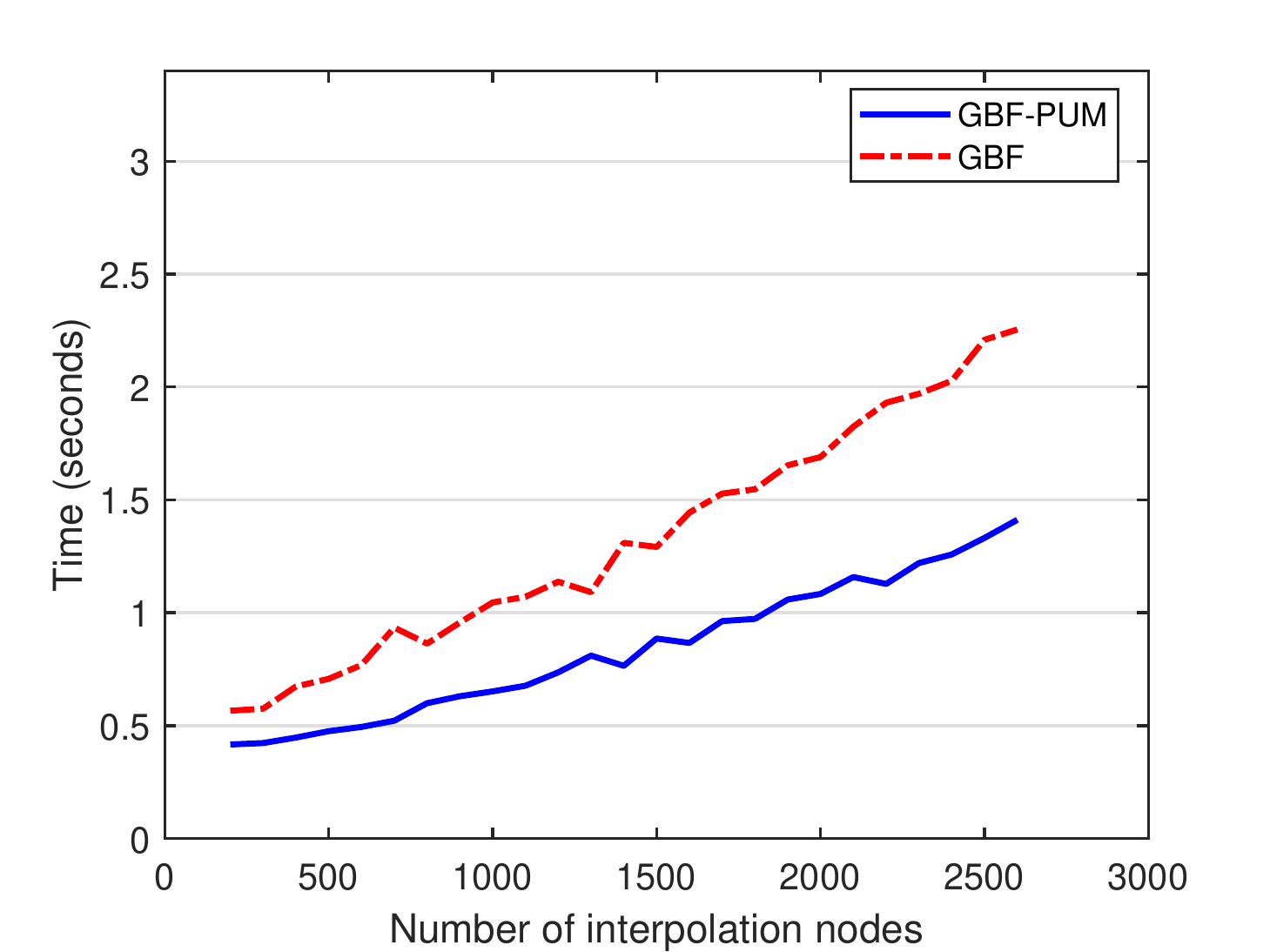}                 
	\caption{RRMSEs and CPU times (in seconds) obtained for the GBF-PUM (with $J = 8$ and $r=12$) and the global GBF scheme in terms of $N$ interpolation nodes. Tests have been carried out by using the variational spline kernel with $\epsilon = 0.001$ and $s=2$.}
	\label{fig2:RRMSE-time}
\end{figure}

These first experimental results indicate that the parameters $J$ (number of clusters) and $r$ (augmentation distance for the subdomains) are two main control parameters for the accuracy and the computational cost of the PUM on graphs. After testing the PUM for a range of different values $J$ and $r$, we made the following observations: (a) the experiments suggest that a good compromise between $J$ and $r$ must be found in order to guarantee both, accuracy and efficiency; (b) improved accuracy results can be obtained if $J$ increases and, simultaneously, $r$ is enlarged too; (c) increasing the parameter $r$ seems to improve accuracy more than increasing $J$. As an enlarged parameter $r$ leads also to a larger computational cost, further investigations are necessary to find an optimal balancing of the different parameters.  

\section{Conclusions and future work}
In this paper, we proposed an efficient way to generate partitions of unity on graphs via greedy-type $J$-center clustering and a domain augmentation procedure. We investigated how the PUM can be used to reconstruct global signals from local GBF approximations on subdomains, and showed that under suitable restrictions on the PUM global error estimates are inherited from local ones. Moreover, some numerical results were presented, showing the cost-efficiency and the accuracy of the method. Work in progress concerns the optimal selection of the parameters in the PUM, namely the number $J$ of the subdomains, the augmentation factor $r$ for the subdomains, and more sophisticated Shepard weights for the partition of unity. The first experiments presented here showed that a proper choice of these parameters highly influences accuracy and efficiency of the computational procedure. A further development step consists in providing more adaptive techniques for the selection of the partitions on the graph in order to tailor the PUM best possibly to the underlying graph topology.

\section*{Acknowledgments}
The first two authors acknowledge support from the Department of Mathematics \lq\lq Giuseppe Peano\rq\rq\ of the University of Torino via Project 2020 \lq\lq Models and numerical methods in approximation, in applied sciences and in life sciences\rq\rq. This work was partially supported by INdAM-GNCS. This research has been accomplished within RITA (Research ITalian network on Approximation).

\end{document}